\documentclass[11pt]{article}

\usepackage{amsfonts}
\usepackage{amsmath}
\usepackage{amssymb}
\usepackage{amsthm}
\usepackage{mathrsfs}
\usepackage{graphicx}
\theoremstyle{plain}
\newtheorem{theorem}{Theorem}
\newtheorem{proposition}[theorem]{Proposition}
\newtheorem{definition}[theorem]{Definition}
\newtheorem{lemma}[theorem]{Lemma}
\newtheorem{corollary}[theorem]{Corollary}

\theoremstyle{definition}
\newtheorem{example}[theorem]{Example}
\newtheorem{remark}[theorem]{Remark}

\numberwithin{exercise}{section}
\numberwithin{equation}{section}
\numberwithin{theorem}{section}
\numberwithin{problem}{section}

\numberwithin{figure}{section}
\DeclareMathOperator{\diag}{diag}
\DeclareMathOperator{\diam}{diam}

\newcommand{\bs}[1]{{\boldsymbol{#1}}}

\newcommand{\R}{\mathbf{R}}
\newcommand{\Z}{\mathbf{Z}}
\newcommand{\N}{\mathbf{N}}

\begin{document}

\title{{On Eigen's quasispecies model, two-valued fitness landscapes, and isometry groups acting on
finite metric spaces}}

\author{Yuri S. Semenov$^{1}$, Artem S. Novozhilov$^{{2},}$\footnote{Corresponding author: artem.novozhilov@ndsu.edu} \\[3mm]
\textit{\normalsize $^\textrm{\emph{1}}$Applied Mathematics--1, Moscow State University of Railway Engineering,}\\[-1mm]\textit{\normalsize Moscow 127994, Russia}\\[2mm]
\textit{\normalsize $^\textrm{\emph{2}}$Department of Mathematics, North Dakota State University, Fargo, ND 58108, USA}}

\date{}

\maketitle

%Abstract goes here
\begin{abstract}
A two-valued fitness landscape is introduced for the classical Eigen's
quasispecies model. This fitness landscape can be considered as a
direct generalization of the so-called single or sharply peaked
landscape. A general, non permutation invariant quasispecies model
is studied, therefore the dimension of the problem is $2^N\times
2^N$, where $N$ is the sequence length. It is shown that if the fitness function is equal to $w+s$
on a $G$-orbit $A$ and is equal to $w$
elsewhere, then the mean population fitness can be found as the
largest root of an algebraic equation of degree at most $N+1$. Here  $G$ is an arbitrary isometry group
acting on the metric space of sequences of zeroes and ones of the length $N$ with the Hamming distance. An
explicit form of this exact algebraic equation is given in terms of the spherical
growth function of the $G$-orbit $A$. Sufficient conditions for
the so-called error threshold for sequences of orbits are
given. Motivated by the analysis of the two-valued fitness landscapes an abstract generalization of Eigen's model is
introduced such that the sequences are identified with the points of a finite metric space $X$ together with a group of isometries acting transitively on $X$. In particular, a simplicial analogue of the original quasispecies model is discussed, which can be considered as a mathematical model of the switching of the antigenic variants for some bacteria.

\paragraph{\small Keywords:} Eigen's quasispecies model, single peaked landscape, mean population fitness, regular polytope, finite metric space,
isometry group
\paragraph{\small AMS Subject Classification:} 15A18; 92D15; 92D25
\end{abstract}

\section{Introduction. Classical quasispecies model}A great deal of research on the border between mathematics and biology was spurred by Eigen's \textit{quasispecies model}, formulated in 1971 in \cite{eigen1971sma}. This model was suggested to describe the replication of prebiotic macromolecules in order to study various aspects of the problem of the origin of life. Independently, an equivalent model was suggested to study the change of frequencies of different genotypes in haploid multi-allele populations under the evolutionary forces of selection and mutation. Standard references to review the classical and recent developments are \cite{baake1999,eigen1988mqs,jainkrug2007,wilke2005quasispecies,schuster2012evolution}. We also refer to the introductory sections in \cite{bratus2013linear,semenov2014,semenov2015} for more details on various issues in the quasispecies theory. In the present work we are mostly concerned with some specific mathematical developments about the model, which can also describe various systems in population biology or chemical kinetics.

We start with formulating the model. We assume that we deal with a population of sequences of the fixed length $N$. Each sequence is composed of zeroes and ones, hence $l:=2^N$ being the total number of different types of sequences. The sequences can reproduce and mutate to each other. We also assume that the reproduction events occur at discrete time moments, and sequence $k$ produces $w_k$ offspring on average with the probabilities $q_{jk}$, where $q_{jk}$ is the probability to produce sequence $j$ by the parent of type $k$. Therefore, $q_{kk}$ is the probability of the error-free reproduction, and $\sum_{j}q_{jk}=1$. Let $\bs p\in\R^l$, $\bs p^\top=(p_0,\ldots,p_{l-1})$ be the vector of frequencies of different types of sequences at the selection-mutation equilibrium. Then it follows from the basic theory (e.g., \cite{burger2000mathematical}) that $\bs p$ can be found as the positive normalized eigenvector of the matrix $\bs{QW}$ corresponding to the dominant eigenvalue $\lambda$, i.e.,
\begin{equation}\label{i:1}
    \bs{QWp}=\lambda \bs p.
\end{equation}
Here $\bs W=\diag(w_0,\ldots,w_{l-1})$ is the matrix describing the \textit{fitness landscape} (note that we count the indices from 0), and $\bs Q=(q_{jk})_{l\times l}$ is the mutation matrix, which is stochastic by definition. At the equilibrium the dominant eigenvalue $\lambda$ is equal to the mean population fitness $\lambda=\overline{w}:=\sum_j w_jp_j$, and the vector $\bs p$ was called the \textit{quasispecies} by Eigen and his co-authors.

The basic mathematical problem, given $\bs W$ and $\bs Q$, is to determine $\overline{w}$ and $\bs p$. This problem turned out to be very nontrivial and required an introduction of intricate methods of statistical physics, careful numerical procedures, and non-elementary mathematical analysis to achieve a partial progress (much more detail can be found in \cite{bratus2013linear,semenov2014,semenov2015} and references therein). No general analytical solution exists. Moreover, even numerically, there are important obstacles to find $\overline{w}$ and/or $\bs p$, most serious of which is the dimensionality of the problem, recall that the matrices have the dimensions $l\times l=2^N\times 2^N$, where $N$ is the sequence length. One particular solution to the problem of dimensionality is to consider very special fitness landscapes, such that the average number of offspring is determined not by the sequence type (which is the ordered list of ones and zeroes) but by the sequence composition (i.e., by the numbers of ones and zeroes in the sequence). Such fitness landscapes are sometimes called \textit{symmetric} or \textit{permutation invariant} and allow to reduce the dimension of the problem from $2^N\times 2^N$ to $(N+1)\times (N+1)$. This worked especially well for the so-called \textit{single peaked fitness landscape} defined by
$$
\bs W=\diag (w+s,w,\ldots,w),\quad w\geq 0,\,s>0,
$$
see \cite{galluccio1997exact,semenov2014,swetina1982self} for additional details. Moreover, most limiting procedures when the sequence length $N$ tends to $\infty$, were applied to the models with the permutation invariant fitness landscapes, e.g., \cite{Baake2007,saakian2006ese,semenov2015}. At the same time it is clear that the assumption that the fitness landscape is permutation invariant should be relaxed at least in some specific biological situations.

Our first goal in this manuscript is to present an efficient method to reduce the dimensionality of the mathematical problem from $2^N$ to $N+1$ for some specific fitness landscapes that generalize the single peaked landscape but are not permutation invariant \textit{sensu} the definition given above. These fitness landscapes still possess a great deal of symmetry but are much more flexible for assigning the fitness values compared to the permutation invariant landscapes. Second, by carefully analyzing the obtained algebraic equation for $\overline{w}$ we are able to give a precise mathematical definition of the threshold-like behavior, which is observed in some quasispecies models \cite{wilke2005quasispecies}. We present sufficient conditions for the model to demonstrate such behavior. The language of the group theory allows us to recast the conditions for the error threshold to occur into the geometric picture of sequences of orbits in the underlying metric space under the action of a given group. Third, motivated by these considerations, we introduce an abstract generalized Eigen quasispecies problem, give several specific examples, and briefly analyze a simplicial analogue of the original quasispecies model. Despite a high level of abstraction of the introduced model, even the simplest mathematical construction describes biologically realistic systems, in particular, the switching of the antigenic variants for some bacteria\footnote{A somewhat modified version of this text appeared in Bull Math Bio, 78(5), 991--1038, 2016, which also includes a nontechnical discussion of the main results. The same discussion can be found also at \textsf{https://anovozhilov.wordpress.com/2016/03/25/eigen-quasispecies-model-and-isometry-groups/}}.

\section{Notation. The reduced problem}\label{sec:1}
In this section we introduce the required notation and also list several facts necessary for our exposition, additional details can be found in \cite{semenov2014}. Recall that $N$ denotes the sequence length and $l=2^N$.

Let $A$ be a non-empty fixed subset of indices:
$A\subseteq\{0,1,\dots,l-1\}$. For some fixed $w\ge 0$,
$s>0$ we consider the two-valued fitness landscapes of the form
$$
w_k=\left\{
\begin{array}{r}
w+s,\quad k\in A\;,\\
w,\quad k\notin A\;.\\
\end{array}
\right.
$$
Thus, the diagonal matrix ${\bs W}$ of fitnesses can be
represented as
\begin{equation}\label{eq:1}
{\bs W}=w {\bs I}+s{\bs E}_A=w {\bs I}+s
\sum_{a\in A}{\bs E}_a,
\end{equation}
$\bs{ I}$ being the identity matrix and ${\bs E}_a$ being the
elementary matrix with the only one nontrivial entry $e_{aa}=1$ on
the main diagonal.

Consider the eigenvalue problem
\begin{equation}\label{eq:2}
{\bs{QW}} \bs p=\overline{w}\, \bs p,
\end{equation}
where $\bs p^{\top}=(p_0,\dots,p_{l-1})\in \R^l$ is the eigenvector of the matrix
$ {\bs{QW}}$ corresponding to the leading (dominant) eigenvalue
$\lambda=\overline{w}$. The vector $\bs p$ is normalized such
that
\begin{equation}\label{eq:3}
\sum_{k=0}^{l-1}p_k=1,\quad p_k> 0\,,
\end{equation}
and hence the equality
\begin{equation}\label{eq:4}
\sum_{k=0}^{l-1}w_kp_k=w+s\sum_{a\in A}p_a= \overline{w}
\end{equation}
holds.

For the following we make an additional assumption that the mutations at different sites of the sequences are independent and the fidelity (i.e., the probability of the error-free reproduction) per site per replication is given by the same constant $0\leq q\leq 1$ for each site. Then
$$
q_{jk}=q^{N-H_{jk}}(1-q)^{H_{jk}}, \quad j,k=0,\ldots,l-1
$$
defines the mutation matrix $\bs Q$. Here $H_{jk}$ is the standard Hamming distance between sequences $j$ and $k$ (i.e., the number of sites at which sequences $j$ and $k$ are different). Note that now both the leading eigenvalue $\lambda$ and the quasispecies $\bs p$ depend on the fitness landscape and, most importantly, on the mutation fidelity $q$, hence we sometimes denote $\overline{w}=\overline{w}(q)$ and $\bs p=\bs p(q)$.

Using the special structure of the mutation matrix $\bs Q$, it can be shown (see, e.g, \cite{semenov2014}) that there exists a non-degenerate matrix $\bs T$ such that
$$
{\bs T}^{-1}=\frac{1}{l}{\bs T},\quad{\bs T}^{-1}{\bs{QT}}=\frac{1}{l}{\bs{ TQT}}=:\bs D,
$$
with
$$
\bs D=\diag\bigl(1,\dots,(2q-1)^{H_{j}},\dots,(2q-1)^{H_{l-1}}\bigr),
$$
where $H_{j}$ is the Hamming norm of the sequence $j$, i.e., the number of ones in this sequence, $H_j:=H_{0j}$; we are using the lexicographical ordering of indices, hence, e.g., $H_0=0$ and $H_{l-1}=N$. Moreover, explicitly matrix $\bs T$ is given through the recursive procedure
$$
\bs T=\bs T_{N},\quad \bs T_{k}=\bs T_1\otimes \bs T_{k-1},\quad k=2,3,\ldots,N,
$$
and
$$
\bs T_1=\begin{bmatrix}
          1 & 1 \\
          1 & -1 \\
        \end{bmatrix}.
$$
Here $\otimes$ denotes the Kronecker product (e.g., \cite{laub2005matrix}).

We write down the indices $a,b$, where
$0\le a\le l-1$, $0\le b\le l-1$, in the binary representation:
$$ a=\alpha_0+\alpha_12+\dots+\alpha_{N-1}2^{N-1}=[\alpha_0,\,\alpha_1,\,\dots\,,\,\alpha_{N-1}]\,,\quad
\alpha_k\in\{0,1\}\,,
$$
$$b=\beta_0+\beta_12+\dots+\beta_{N-1}2^{N-1}=[\beta_0,\,\beta_1,\,\dots\,,\,\beta_{N-1}]\,,\quad
\beta_k\in\{0,1\}\,.
$$

One additional property of $\bs T$ that we will require in the following is given by the following lemma.

\begin{lemma}\label{l:1} Let ${\bs T}=(t_{ab})_{l\times l}$ be the transition
matrix defined above, $P_a(z)=\sum\limits_{k=0}^{l-1}t_{ka}z^k $ be
the generating polynomial of the $a$-th column of $\bs T$. Then
\begin{equation}\label{eq:5}
P_a(z)=\sum\limits_{k=0}^{l-1}t_{ka}z^k=\prod_{i=0}^{N-1}\left(1+(-1)^{\alpha_i}z^{2^i}\right)\;.
\end{equation}
Moreover,
\begin{equation}\label{eq:6}
t_{ab}=t_{ba}=(-1)^{\left<a,b\right>}\;,\quad
\left<a,b\right>:=\sum\limits_{k=0}^{N-1}\alpha_k\beta_k\bmod 2.
\end{equation}
\end{lemma}
\begin{proof} We prove
the formulas \eqref{eq:5} and \eqref{eq:6} by the induction on $N$. Indeed,
according to the definition of $\bs T$ we have the Kronecker product $${\bs T}={\bs T}_{N}={\bs T}_1\otimes {\bs
T}_{N-1}=\left[\begin{array}{rr}1&1\\1&-1\end{array}\right]\otimes
{\bs T}_{N-1}=\left[\begin{array}{rr}{\bs T}_{N-1}&{\bs
T}_{N-1}\\{\bs T}_{N-1}&-{\bs T}_{N-1}\end{array}\right].$$ Let
us represent $a=a_N=a_{N-1}+\alpha_{N-1} 2^{N-1}$. The block form
of ${\bs T}={\bs T}_{N}$ implies the equality
$P_a(z)=P_{a_{N-1}}(z)(1+(-1)^{\alpha_{N-1}}z)$ and \eqref{eq:5} by
induction.

Consider again the representations $a=a_N=a_{N-1}+\alpha_{N-1}
2^{N-1}$, $b=b_N=b_{N-1}+\beta_{N-1}2^{N-1}$. It follows
from the above form of the matrix ${\bs T}$  that
$$
t_{ab}=({\bs T}_N)_{ab} =({\bs
T}_{N-1})_{a_{N-1}b_{N-1}}(-1)^{\alpha_{N-1}\beta_{N-1}}\;.
$$
The induction on $N$ completes the proof of \eqref{eq:6}.
\end{proof}

Let us now return to the problem \eqref{eq:2}. We have
$$
{\bs T}^{-1}{\bs{ QT}} {\bs T}^{-1}{\bs{ WT}} {\bs
T}^{-1}\bs p=\overline{w}\, {\bs T}^{-1}\bs p\;,
$$
or, in view of \eqref{eq:1},
\begin{equation}\label{eq:7}
{\bs D} \left(w{\bs I}+s\sum_{a\in A}{\bs T}^{-1}{\bs
E}_a{\bs T}\right) {\bs T}^{-1}\bs p=\overline{w}\, {\bs
T}^{-1}\bs p\;,
\end{equation}
which yields,  after some rearrangement,
\begin{equation}\label{eq:8}(\overline{w}{\bs I}-w{\bs D}) {\bs
T}^{-1}\bs p=\frac{s}{l}\sum_{a\in A}{\bs{ DT}}{\bs E}_a\bs p\;.
\end{equation}

Let $\bs x:={\bs T}^{-1}\bs p$, $\bs p={\bs T}\bs x$. Then \eqref{eq:8} implies
\begin{equation}\label{eq:9}
\bs x=\frac{s}{l}\sum_{a\in A}(\overline{w}{\bs I}-w{\bs D})^{-1}
{\bs{ DT}}{\bs E}_a\bs p\;,
\end{equation}
or, in coordinates,
\begin{equation}\label{eq:10}
x_k=\frac{s}{l}\sum_{a\in A}\frac{(2q-1)^{H_k}
t_{ka}p_a}{\overline{w}-w(2q-1)^{H_k}}\,,\quad
k=0,\dots,l-1.
\end{equation}
Since $\bs p={\bs T}\bs x$, then we get
\begin{equation}\label{eq:11}
p_b=\sum_{k=0}^{l-1}t_{bk}x_k=\frac{s}{l}\sum_{a\in
A}\sum_{k=0}^{l-1}\frac{(2q-1)^{H_k}
t_{bk}t_{ak}}{\overline{w}-w(2q-1)^{H_k}}p_a\,.
\end{equation}

Note that only the components $p_a$,
where $a\in A$, are involved in the right-hand side of \eqref{eq:11}. We
can omit the components $p_b$ for $b\notin A$ and obtain the
``reduced'' column-vector $\bs p_A=(p_a)$, $a\in A$. Considering only
$a\in A$ we can rewrite \eqref{eq:11} as
\begin{equation}\label{eq:12}
\bs p_A=\bs M\bs p_A\,,
\end{equation}
where $\bs M=(m_{ab})_{r\times r}$ is the square matrix of the order $r=|A|$
with the entries
\begin{equation}\label{eq:13}
m_{ab}=m_{ba}=\frac{s}{l}\sum_{k=0}^{l-1}\frac{(2q-1)^{H_k}
t_{ak}t_{bk}}{\overline{w}-w(2q-1)^{H_k}}=\frac{s}{l}\sum_{k=0}^{l-1}\frac{(2q-1)^{H_k}
(-1)^{\left<a,k\right>+\left<b,k\right>}}{\overline{w}-w(2q-1)^{H_k}}\,
\end{equation}
in view of \eqref{eq:11} and Lemma \ref{l:1}. The equality \eqref{eq:12} means that the
reduced vector $\bs p_A$ is an eigenvector of $\bs M$ corresponding to the
eigenvalue $\lambda=1$.

We consider $\overline{w}$ in \eqref{eq:13} as a parameter. It
follows from \eqref{eq:4} that $\overline{w}$ depends only on $p_a$, $a\in
A$, that is, on the reduced vector $\bs p_A$. The original eigenvector $\bs p$ can be
reconstructed from $\bs p_A$ with the help of \eqref{eq:11} if $\bs p_A$ is known.
Therefore, for the introduced special two-valued fitness landscapes, instead of the original problem \eqref{eq:2}, we can consider the problem to find the reduced eigenvector $\bs p_A$ satisfying \eqref{eq:12} and corresponding to the eigenvalue $\lambda=1$ of the matrix $\bs M$ defined in \eqref{eq:13}. Since \eqref{eq:12} defines $p_a,\,a\in A$ in terms of $\overline{w}$, then, finally, formula~\eqref{eq:4} can be used to determine $\overline{w}$ implicitly in terms of the system parameters $w,s$ and $q$.

%\begin{quote}\it To find the eigenvector $\bs p_A$ satisfying \eqref{eq:12} and corresponding to the eigenvalue $\lambda=1$ of the matrix $\bs M$ defined in \eqref{eq:13}. The mean fitness $\overline{w}=\overline{w}(q)$ then satisfies the formula~\eqref{eq:4}, which defines $\overline{w}$ implicitly in terms of $w,s$ and $q$.
%\end{quote}

To conclude, we remark that the eigenvalue $\overline{w}$ can be also
found from the equation
\begin{equation}\label{eq:15}\det({\bs M}-{\bs I})=0\,,
\end{equation}
but in general it is not easier than to solve the original
problem \eqref{eq:2}. For the single peaked landscapes (i.e., when $A$ consists
of a single element) the corresponding equation was obtained and
investigated in \cite{semenov2014}. In the next section we propose a different
approach that can be further elaborated on for some special cases.

\section{Equation for the leading eigenvalue $\overline{w}$}\label{sec:2}
In this section we show that, using the preliminary analysis from the previous section, it is possible to find an algebraic equation for the eigenvalue $\overline{w}$ under some additional symmetry requirements on the set $A$, and this equation is of degree at most $N+1$.

First we transform \eqref{eq:13} in the following way:
\begin{equation}\label{eq:16}
\begin{split}
m_{ab}&=\frac{s}{l}\sum_{k=0}^{l-1}\frac{(2q-1)^{H_k}t_{ak}t_{bk}}{\overline{w}-w(2q-1)^{H_k}}=\frac{s}{l\overline{w}}\sum_{k=0}^{l-1}\frac{(2q-1)^{H_k}t_{ak}t_{bk}}{1-\frac{w}{\overline{w}}(2q-1)^{H_k}}\\
&=\frac{s}{l\overline{w}}\sum_{k=0}^{l-1}t_{ak}t_{bk}\sum_{c=0}^{\infty}\left(\frac{w}{\overline{w}}\right)^c (2q-1)^{(c+1)H_k}=\frac{s}{l\overline{w}}\sum_{c=0}^{\infty}\left(\frac{w}{\overline{w}}\right)^c\sum_{k=0}^{l-1}(2q-1)^{(c+1)H_k}t_{ak}t_{bk}\;.
\end{split}
\end{equation}

\begin{lemma}\label{l:2} We have the following factorization:
\begin{equation}\label{eq:17}\sum_{k=0}^{l-1}z^{H_k} t_{ak}t_{bk}=
(1-z)^{H_{ab}}(1+z)^{N-H_{ab}}\;.
\end{equation}
\end{lemma}
\begin{proof}
It is straightforward to see that $\sum\limits_{k=0}^{l-1}z^{H_k}
t_{ak}t_{bk}= \sum\limits_{k=0}^{l-1}t_{ak}z^{H_k} t_{kb}$ is the
entry $z_{ab}$ of the matrix
$$
{\bs Z}:={\bs T}\diag(1,\dots,z^{H_k},\dots,z^N){\bs T}=2^N {\bs T}\diag (1,\dots,z^{H_k},\dots,z^N){\bs T}^{-1}.
$$
It follows from the properties of $\bs T$ that
$$
z_{ab}=2^N(1-q)^{H_{ab}}q^{N-H_{ab}}=(1-z)^{H_{ab}}(1+z)^{N-H_{ab}},\quad
\mbox{where}\; q=\frac{1+z}{2}\;.
$$
Thus, the lemma is proved. \end{proof}

\medskip Applying Lemma \ref{l:2} to \eqref{eq:16} we get
\begin{equation}\label{eq:18}
m_{ab}=\frac{s}{l\overline{w}}
\sum_{c=0}^{\infty}\left(\frac{w}{\overline{w}}\right)^c
\left(1-(2q-1)^{c+1}\right)^{H_{ab}}\left(1+(2q-1)^{c+1}\right)^{N-H_{ab}}.\end{equation}

We have the following consequence of \eqref{eq:12}:
\begin{equation}\label{eq:19}
\sum_{b\in A} p_b= \sum_{b\in A}\sum_{a\in A} m_{ba}p_a=\sum_{a\in
A}p_a \sum_{b\in A} m_{ba}.
\end{equation}

Now we introduce a \textit{key} assumption that will allow us to simplify the analysis.

We assume that the sum $\sum\limits_{b\in A}m_{ab}$ does not depend
on $a\in A$. Then it follows from \eqref{eq:19} that $\sum\limits_{b\in A}
m_{ba}=1$ for each $a\in A$. In view of \eqref{eq:18} it implies that
\begin{equation}\label{eq:20}
\frac{l\overline{w}}{s}=
\sum_{c=0}^{\infty}\left(\frac{w}{\overline{w}}\right)^c
\sum\limits_{b\in A}
\left(1-(2q-1)^{c+1}\right)^{H_{ab}}\left(1+(2q-1)^{c+1}\right)^{N-H_{ab}}\,,
\end{equation}
if the inner sum does not depend on $a\in A$.

The main question is when our key assumption holds. We present some sufficient conditions for this, and hence for the equation \eqref{eq:20}.

We refer to the well known geometric interpretation of the metric
space $V=V_N=\{0,1\}^N$ with the Hamming distance. Consider
1-skeleton of the $N$-dimensional cube $[0,1]^N$ with the set of
vertices $V$. The vertices $a$ and $b$ are connected by the
(unique) edge $e_{ab}$ if $H_{ab}=1$. The Hamming distance between
vertices $u$ and $v$ is the length of a shortest path connecting
these vertices, that is, the number of edges in this path. The set $V$, due to the binary representation
$$a=\alpha_0+\alpha_1
2+\dots+\alpha_{N-1}2^{N-1}=[\alpha_0,\,\alpha_1,\,\dots\,,\,\alpha_{N-1}]\,,\quad
\alpha_k\in\{0,1\}\,,$$ can be identified with the set of
indices $X=X_N=\{0,1,\dots,2^N-1\}$ with the Hamming distance. In what follows we will usually make no difference between metric
spaces $V$ and $X$.

We note that the group ${\rm Iso}(X_N)$ of all isometries of $X_N$, acting on the set $X_N$, is also known as the \textit{Weyl group} $W_N$ of order $2^NN!$ of the root system of type $B_N$ (or $C_N$, see, e.g., \cite{bourbaki22001lie}).

\begin{proposition}\label{pr:1} Let $G$ be a group
that acts on the metric space $X$ by isometries \emph{(}i.e.,
$G\leqslant{\rm Iso}(X)$\emph{)} and let $A$ be a $G$-orbit. Then
the equality \eqref{eq:20} holds.
\end{proposition}

\begin{proof} Since $G$ acts transitively on $A$ and preserves the Hamming
distance $H_{ab}$, the inner sum in \eqref{eq:20} does not depend on $a\in
A$.
\end{proof}

Now we can state the following basic result.
\begin{corollary}\label{corr:2:1}Under the conditions of Proposition \ref{pr:1} the eigenvalue $\overline{w}=\overline{w}(q)$ of \eqref{eq:2} is a root of an algebraic equation \emph{(}with the coefficients depending on $q$\emph{)} of degree at most $N+1$.
\end{corollary}
\begin{proof}
Consider the polynomial (see \eqref{eq:17} and Lemma \ref{l:2})
\begin{equation}\label{eq:21}F_A(z):=\frac{1}{2^N}\sum_{b\in A}(1-z)^{H_{ab}}(1+z)^{N-H_{ab}}\,.
\end{equation}
We have
\begin{equation}\label{eq:22}
F_A(z)>0,\;-1<z\leq 1\,,\quad F_A(0)=\frac{|A|}{2^N}\;,\quad F_A(1)=1.
\end{equation}
Since $0\leq H_{ab}\leq N$, we can rewrite \eqref{eq:21} as
\begin{equation}\label{eq:23}
F_A(z)=\frac{1}{2^N}\sum_{d=0}^N f_d\,(1-z)^{d}(1+z)^{N-d}
,
\end{equation}
where $f_d=\#\{b\in A\,|\,H_{ab}=d\}$. Applying the
binomial expansion to \eqref{eq:23} we get
\begin{equation}\label{eq:24}
F_A(z)=\sum_{d=0}^N h_d\,z^d\;,\quad
h_d=\frac{1}{2^N}\sum_{j=0}^d(-1)^j f_j{\binom{d}{j}}{\binom{N-d}{d-j}}.
\end{equation}

%\medskip\noindent{\bf Conjecture 2.1.} For any $G$-orbit $A$
%all the coefficients $h_d$ in (24) are non-negative.

%It follows in view of the second equality (22) that if the
%conjecture is valid then $F_A(z)$ is the generating function of
%some distribution. Which one?

With the introduced notation equation \eqref{eq:20} reads
\begin{equation}\label{eq:25}
\frac{\overline{w}}{s}=
\sum_{c=0}^{\infty}\left(\frac{w}{\overline{w}}\right)^c
F_A((2q-1)^{c+1})\,,
\end{equation}
or
$$
\frac{\overline{w}}{s}= \sum_{d=0}^N h_d(2q-1)^d
\sum_{c=0}^{\infty}\left(\frac{w}{\overline{w}}\right)^c(2q-1)^{cd}=\sum_{d=0}^N
\frac{h_d(2q-1)^d\overline{w}}{\overline{w}-w(2q-1)^d}\,.
$$
Finally, the last equation can be transformed into
\begin{equation}\label{eq:26}
\sum_{d=0}^N
\frac{h_d(2q-1)^d}{\overline{w}-w(2q-1)^d}=\frac{1}{s}
\end{equation}
with the rational coefficients $h_d$ defined in \eqref{eq:24}.
\end{proof}

We have the following corollary that allows us to reduce the number of computations in some special cases.
\begin{corollary}\label{corr:2:3} Let $\Gamma=\Gamma_N={\rm Iso}(X_N)$ be
the group of all isometries {\emph{(}}of order $2^N N!${\emph{)}} acting on the metric
space $X=X_N=\{0,1,\dots,2^N-1\}$ with the Hamming distance and let
$\gamma A$ be the image of $A$ under the {\emph{(}}left{\emph{)}} action of
$\gamma\in \Gamma$. Then the equations \eqref{eq:20}, \eqref{eq:25}, \eqref{eq:26}, and \eqref{eq:27} are the same for
$A$ and $\gamma A$.
\end{corollary}
\begin{proof}
The equations \eqref{eq:20}, \eqref{eq:25}, \eqref{eq:26}, and \eqref{eq:27} were obtained only on the ground of the metric properties of $A$. Note that $A$ is a $G$-orbit if and only if $\gamma A$ is a $\gamma
G\gamma^{-1}$-orbit. Thus, the (left-)acting group $G$ should be
substituted by the conjugate $\gamma G\gamma^{-1}$.
\end{proof}

The case $w=0$ corresponds to the lethal mutations. In particular, we have
\begin{corollary} If $w=0$ then we get the
following polynomial expression for the leading eigenvalue, where $a\in
A$ may be chosen arbitrarily in the $G$-orbit $A$:
\begin{equation}\label{eq:27}
\overline{w}=s \sum\limits_{b\in A}
(1-q)^{H_{ab}}q^{N-H_{ab}}=sF_A(2q-1)=s\sum_{d=0}^N
h_d(2q-1)^d\,.
\end{equation}
\end{corollary}

\section{Examples and applications}\label{sec:3}
In this section we consider several simple examples of the two-valued fitness landscapes and apply the obtained equation for the leading eigenvalue. The examples we consider are mostly based on various subgroups of the symmetric group $S_N$.

The symmetric group $G=S_N$ acts on the metric
space $X=X_N=\{0,1,\ldots,2^N-1\}$ with the Hamming distance by
isometries. To wit, let $\sigma\in S_N$. Then
$$\sigma(a)=\sigma[\alpha_0,\,\alpha_1,\ldots,\alpha_{N-1}]=
[\alpha_{\sigma^{-1}(0)},\,\alpha_{\sigma^{-1}(1)},\ldots,\alpha_{\sigma^{-1}(N-1)}],\quad
\alpha_k\in\{0,1\}\,.$$ Note that $G=S_N$ is a proper subgroup
of $\Gamma={\rm Iso}(X_N)$. The latter is of the order $2^N N!$ and
contains also the elements that correspond to reflections of the
$N$-dimensional cube $[0,1]^N$, see, e.g., Example \ref{ex:3:2} below.

The $S_N$-orbits are the subsets of
$$
A_p=\{a\in X\,|\,H_a=p\}\;,\quad p=0,1,\dots,N\,.
$$

\begin{example}[General permutation invariant fitness landscapes]\label{ex:3:1}
Recall that we defined the permutation invariant fitness landscape to be a diagonal matrix $\bs W$ such that the elements on the main diagonal are $w_j=w_{H_j},$ i.e., the fitness of the sequence $j$ depends only on the total number of ones in this sequence. To satisfy this definition the orbit for the two-valued fitness landscape must coincide with one of $A_p$ defined above.

We can consider only the case $2p\le N$. Indeed,
let $\gamma(a)=a^*=l-1-a$ be the index conjugate to $a$. The
conjugation $\gamma$ is an involution in $\Gamma$. The binary representation of $a^*$ differs at each position from
that of $a$. Then $H_{a^*}=N-H_a$ and $a\in A_p\Leftrightarrow
a^*\in A_p^*=A_{N-p}$. In other words, according to  Corollary \ref{corr:2:3},
 the equations \eqref{eq:20}, \eqref{eq:25}, \eqref{eq:26}, and \eqref{eq:27} for $A_p$ and $A_p^*=A_{N-p}$ are the same.

To obtain an equation for $\overline{w}$ we will need an auxiliary
\begin{lemma}\label{l:3:1} For $a,b\in A=A_p$ the distance $H_{ab}$ is even.
Moreover, for each $k=0,1,\dots,p\,$
$$
\#\{b\in A_p\,|\,H_{ab}=2k\}={\binom{p}{k}}{\binom{N-p}{k}}.
$$
\end{lemma}

\begin{proof} If $H_a$ and $H_b$ have the same parity, in particular, coincide then $H_{ab}$ is even, hence
$H_{ab}=2k$.

The binary representations
$$ a=[\alpha_0,\,\alpha_1,\,\dots\,,\,\alpha_{N-1}]\,,
\quad b=[\beta_0,\,\beta_1,\,\dots\,,\,\beta_{N-1}]\,,\quad
\alpha_j,\beta_j\in\{0,1\}\,$$ differ at exactly $2k$ positions.
Thus, in order to obtain the binary representation of $b$ from
that of $a$ we need to substitute exactly $k$ ones in
$[\alpha_0,\,\alpha_1,\,\dots\,,\,\alpha_{N-1}]$ by zeroes and
exactly $k$ zeroes in
$[\alpha_0,\,\alpha_1,\,\dots\,,\,\alpha_{N-1}]$ by ones since the
total number of ones  in both representations of $a$, $b$ is equal
to $p$, $H_a=H_b=p$. There are ${\binom{p}{k}}{\binom{N-p}{k}}$ variants of such substitutions.
\end{proof}

Lemma \ref{l:3:1} applied to \eqref{eq:20} yields
\begin{equation}\label{eq:28}
\frac{\overline{w}}{s}=
\frac{1}{2^N}\sum_{c=0}^{\infty}\left(\frac{w}{\overline{w}}\right)^c
\sum\limits_{k=0}^p {\binom{p}{k}}{\binom{N-p}{k}}\left(1-(2q-1)^{c+1}\right)^{2k}\left(1+(2q-1)^{c+1}\right)^{N-2k}.
\end{equation}
We note that in \eqref{eq:28} we can disregard the restriction $2p\le N$. The
polynomial
\begin{equation}\label{eq:29}F_{A_p}(z)=\frac{1}{2^N}\sum\limits_{k=0}^p
{\binom{p}{k}}{\binom{N-p}{k}}(1-z)^{2k}(1+z)^{N-2k}=\sum\limits_{d=0}^N h_dz^d
\end{equation}
of degree $N$ satisfies the conditions \eqref{eq:22}. Moreover, $h_d=h_{N-d}$.

Therefore, we conclude that for the permutation invariant fitness landscapes we obtained the explicit equation \eqref{eq:26} for the leading eigenvalue $\overline{w}$ with $h_d$ defined in \eqref{eq:29}. While it is a common wisdom that the dimensionality of the quasispecies problem for the permutation invariant fitness landscapes can be reduced to $N+1$ from $2^N$, the explicit equation to determine the leading eigenvalue $\overline{w}$ is, to the best of our knowledge, new.
\end{example}
\begin{remark}For the permutation invariant fitness landscapes arguably the most transparent and efficient way to analyze the problem is to invoke the so-called maximum principle \cite{Baake2007,saakian2006ese,wolff2009robustness}, therefore, first several examples in this section should be mostly considered as an illustration of the suggested technique. Nevertheless, the results we present are exact, contrary to the approximate nature of the maximum principle, for which also some technical conditions on the fitness landscape must be imposed (without these conditions the maximum principle can lead to incorrect conclusions, e.g., \cite{semenov2015,wolff2009robustness}). Our equations work for any fitness landscape and therefore are of interest on their own.
\end{remark}
In what follows we consider several special cases of the previous example.
\begin{example}[Single peaked landscape]\label{ex:3:1:1} Let $p=0$ in the previous example. Then we deal with the classical single peaked fitness landscape. The equation for $\overline{w}$ was studied in great details in \cite{semenov2014}. We would like to mention that in view of Corollary \ref{corr:2:3}, since the group of isometries $\Gamma$ acts transitively on
the set $X$ of indices then each equation \eqref{eq:26} for the single peaked landscape $A=\{a\}$ is the same.
Consequently, for the leading eigenvalue we can consider the basic case $A_0=\{0\}$, that is, the single peak at $w_0=w+s$.

We can also consider the trivial group $G=\{1\}$ acting on $X$ in order to treat the same case.

The polynomial \eqref{eq:23} becomes
$$
F_{A_0}(z)=\frac{1}{2^N}(1+z)^N= \frac{1}{2^N}\sum\limits_{d=0}^N
{\binom{N}{d}}z^d.
$$
Hence, the equation \eqref{eq:26} reads
\begin{equation}\label{eq:30}
\frac{1}{2^N} \sum_{d=0}^{N} {\binom{N}{d}}\frac{(2q-1)^{d}}{\overline{w}-w(2q-1)^{d}}=\frac{1}{s}\,.
\end{equation}
A very similar expression for a slightly different model was obtained originally in \cite{galluccio1997exact}.
\end{example}

\begin{example}\label{ex:3:1:2}In the previous notation let
$A=A_1=\{1,2,4,8,\dots,2^{N-1}\}$. Now, for
$a,b\in A$
$$
H_{ab}=\left\{
\begin{array}{r}
0,\quad a=b\;,\\
2,\quad a\ne b\;.\\
\end{array}
\right.
$$
The calculation of the polynomial \eqref{eq:23} yields
$$
F_{A_1}(z)=\frac{1}{2^N}\left((1+z)^N+(N-1)(1-z)^2(1+z)^{N-2}\right)=
\frac{1}{2^N}\sum\limits_{d=0}^N \frac{(N-2d)^2}{N}{\binom{N}{d}}z^d.
$$
Hence \eqref{eq:26} transforms into
\begin{equation}\label{eq:31}
\frac{1}{2^N}\sum_{d=0}^{N} \frac{(N-2d)^2}{N}{\binom{N}{d}}\frac{(2q-1)^{d}}{\overline{w}-w(2q-1)^{d}}=\frac{1}{s}\,.
\end{equation}
\end{example}

\begin{example}\label{ex:3:1:3} Let $N=2n$ be an even number and let $A=A_n$. Applying Lemma \ref{l:3:1} to \eqref{eq:20} we find
\begin{equation}\label{eq:32}
\frac{\overline{w}}{s}=
\frac{1}{2^N}\sum_{c=0}^{\infty}\left(\frac{w}{\overline{w}}\right)^c
\sum\limits_{k=0}^n {\binom{n}{k}}^2\left(1-(2q-1)^{c+1}\right)^{2k}\left(1+(2q-1)^{c+1}\right)^{2n-2k}.
\end{equation}
%Instead of the exact equation \eqref{eq:26} we prefer here to find an
%asymptotic version of \eqref{eq:32} as $n\to \infty$.
%
%It is known that if
%$-1<t<1$ then
%$$
%\lim_{n\to\infty}\frac{1}{{\binom{2n}{n}}}\sum\limits_{k=0}^n
%{\binom{n}{k}}^2(1-t)^{2k}(1+t)^{2n-2k}=\frac{1}{\sqrt{1-t^2}}\,.
%$$
%Thus, for $n$ large enough  \eqref{eq:32} transforms into
%$$
%\frac{\overline{w}}{s}\approx \frac{{\binom{2n}{n}}}{2^{2n}}\sum_{c=0}^{\infty}\left(\frac{w}{\overline{w}}\right)^c
%\frac{1}{\sqrt{1-(2q-1)^{2(c+1)}}}\,,
%$$
%or, by Stirling's formula,
%$$
%\frac{\overline{w}}{s}\approx \frac{1}{\sqrt{\pi
%n}}\sum_{c=0}^{\infty}\left(\frac{w}{\overline{w}}\right)^c
%\frac{1}{\sqrt{1-(2q-1)^{2(c+1)}}}\,.
%$$
%Finally, multiplying by $w\sqrt{\pi n}$ and dividing by
%$\overline{w}$ we get
%\begin{equation}\label{eq:33}
%\frac{w\sqrt{\pi n}}{s}\approx
%\sum_{m=1}^{\infty}\left(\frac{w}{\overline{w}}\right)^m
%\frac{1}{\sqrt{1-(2q-1)^{2m}}}\,,\quad n\gg 1\,.
%\end{equation}
%
\end{example}

\begin{example}[Antipodal landscape]\label{ex:3:2} Consider the set $A=\{a,a^*\}$, where, as before, $a^*$ is the conjugate index, $a^*=l-1-a$.
Let $G=\{1,g\}$ be the group of order 2 whose nontrivial element
(involution) $g$ maps each $a\in X$ to the conjugate $a^*$.
Thus, the set $A=\{a,a^*\}$ is a $G$-orbit. In view of Proposition
\ref{pr:1}, since $H_{aa}=0$ and $H_{aa^*}=N$ then the equation \eqref{eq:20}
reads
\begin{equation}\label{eq:34}
\frac{\overline{w}}{s}=
\frac{1}{2^N}\sum_{c=0}^{\infty}\left(\frac{w}{\overline{w}}\right)^c
\left(
\left(1-(2q-1)^{c+1}\right)^{N}+\left(1+(2q-1)^{c+1}\right)^{N}\right)\,.
\end{equation}
In this case the polynomial \eqref{eq:23} takes the form ($\lfloor\,\cdot\,\rfloor$
stands for the integer part):
$$
F_{A}(z)=\frac{1}{2^N}((1+z)^N+(1-z)^N)= \frac{1}{2^{N-1}}
\sum_{d=0}^{ \lfloor N/2\rfloor} {\binom{N}{2d}}z^d\;.
$$
Hence the equation \eqref{eq:26} is of degree $\lfloor N/2\rfloor+1$:
\begin{equation}\label{eq:35}
\frac{1}{2^{N-1}}\sum_{d=0}^{ \lfloor N/2 \rfloor} {\binom{N}{2d}}\frac{(2q-1)^{2d}}{\overline{w}-w(2q-1)^{2d}}=\frac{1}{s}\,.
\end{equation}
\end{example}

\begin{example}[Quaternion landscape]\label{ex:3:3}
According to a well known theorem of Cayley each (finite) group
$G$ is a permutation group (which acts on itself by, for instance,
left shifts). It follows that each finite group $G$ can be
embedded into symmetric group $S_{n}$, $n=|G|$. Since
$S_n$ acts on the set of indices $X=X_n$ then we can find many
$G$-orbits restricting on $G$ the canonical action of $S_n$ on
$X_n$ (see the beginning of this section). Moreover, since there
are standard embeddings $S_n\to S_{n+1}\to S_{n+2}\to\dots$ there
is no problem to construct the action of any finite group $G$ on
the set $X_N$ for $N\ge n$. This gives us a virtually unlimited list of the two-valued fitness landscapes, which are not permutation invariant.

For instance, let
$$G=Q_8=\{\pm 1,\,\pm i,\,\pm j,\,\pm
k\,|\,i^2=j^2=k^2=-1\,,\;ij=k,\,jk=i,ki=j\}$$ ($-1$ commutes with
each element of $Q_8$) be the classical quaternion group of order
8. The embedding $Q_8\to S_8$ is chosen so that $i\to
(0213)(4657)$, $j\to(0415)(2736)$.

Consider a $G$-orbit, say,
$$A=\{7,11,13,14,112,176,208,224\}\subset X_N\,,\quad N\geq 8.$$
Direct calculations yield the polynomial \eqref{eq:23}
\begin{equation}\label{eq:36}
F_{A,N}(z)=\frac{1}{2^N}((1+z)^N+3(1-z)^2(1+z)^{N-2}+4(1-z)^6(1+z)^{N-6})\,.
\end{equation}
For $N=8$ we have
\begin{equation}\label{eq:36a}
F_{A,8}(z)=\frac{1}{2^8}((1+z)^8+3(1-z)^2(1+z)^6+4(1-z)^6(1+z)^2)=$$$$=
\frac{1}{64}(2+z+14z^2+15z^3+15z^5+14z^6+z^7+2z^8)\,.
\end{equation}
%Note that Conjecture 2.1 is valid for this case. Moreover, the
%coefficients are symmetric, $h_d=h_{8-d}$  (is that true for any
%subgroup $G<S_N$ and any $G$-orbit $A\subset X_N$?
Finally, we obtain the following form of \eqref{eq:26}
\begin{equation}\label{eq:37}
\sum_{d=0}^8
\frac{R_d(2q-1)^d}{\overline{w}-w(2q-1)^d}=\frac{64}{s}\;,
\end{equation}
where $R_0=R_8=2$, $R_1=R_7=1$, $R_2=R_6=14$, $R_3=R_5=15$,
$R_4=0$.

Examples of calculating $\overline{w}$ are given in Fig.~\ref{fig:1}, where the case $N=8$ was also checked numerically using the full matrix $\bs{QW}$.
\begin{figure}[!th]
\centering
\includegraphics[width=0.95\textwidth]{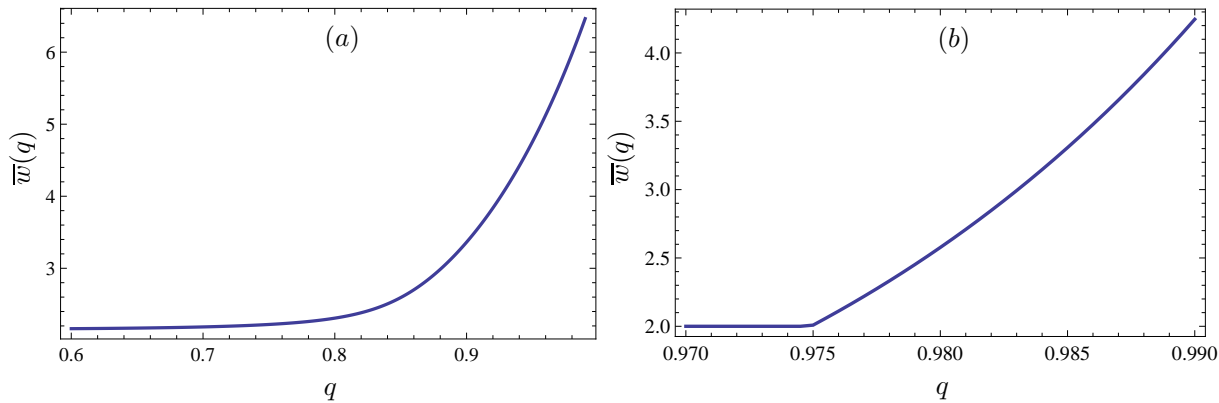}
\caption{The leading eigenvalue $\overline{w}$ depending on the fidelity $q$ for the two-valued fitness landscape with $w=2,\,s=5$ and the set $A$ as in Example \ref{ex:3:3}. $(a)$ $N=8$ (this case was also checked numerically, using the full matrix $\bs{QW}$), $(b)$ $N=50$.}\label{fig:1}
\end{figure}
\end{example}

%\section{On the case $w=0\$ of "mortally dangerous mutations"{}  and on the error threshold.}

\begin{example}[Lethal mutations]\label{ex:3:8}If $w=0$ the calculations can be significantly simplified (see \eqref{eq:27}).
Moreover, we can find not only the polynomial expression
\eqref{eq:27} for the leading eigenvalue provided $A$ is a
$G$-orbit, $G\leqslant{\rm Iso}(X)$, but we can find the
eigenvector $\bs p$ (the quasispecies distribution) as well.

On substituting $w=0$ into \eqref{eq:1} we obtain
\begin{equation}\label{eq:38}
{\bs W}=s {\bs E}_A=s \sum_{a\in A}{\bs E}_a.
\end{equation}
Here ${\bs E}_A$ is the diagonal matrix corresponding to the
projection on the orbit $A$, ${\bs E}_A^2={\bs E}_A$. The problem
\eqref{eq:2} can be transformed now as follows:
\begin{equation}\label{eq:39}
s{\bs Q} {\bs E}_A\bs p=\overline{w} \left({\bs E}_A\bs p+(\bs p-{\bs
E}_A\bs p)\right)\;\quad(=\overline{w}\,\bs p).
\end{equation}
Multiplying \eqref{eq:39} from the left by ${\bs E}_A$ and taking into
account that ${\bs E}_A$ is a projection matrix we find
\begin{equation}\label{eq:40}
{\bs E}_A{\bs Q}  {\bs E}_A\bs p=\frac{\overline{w}}{s}\, {\bs
E}_A\bs p.
\end{equation}
Note that if we omit zeroes in ${\bs E}_A\bs p$, we obtain the reduced
vector $\bs p_A$ introduced in Section \ref{sec:1}.

Direct calculations and Lemma \ref{l:2} show that if we take
\begin{equation}\label{eq:41}
{\bs
E}_A\bs p=\theta(0,\ldots,0,1,0,\dots,0,1,0,\ldots),\quad
\theta>0,
\end{equation}
where the ones stand only for the indices $a\in A$, we
get
\begin{equation}\label{eq:42}
{\bs E}_A{\bs Q}  {\bs E}_A\bs p=F_A(2q-1)\, {\bs
E}_A\bs p.
\end{equation}
Let us compare \eqref{eq:40} and \eqref{eq:42}. In view of \eqref{eq:27}  we conclude
that the vector $\bs p$ satisfying \eqref{eq:41} is a solution of the problem \eqref{eq:2}
provided $\overline{w}=sF_A(2q-1)$ (possibly not unique).

The equality \eqref{eq:39} implies, regardless of $\theta$, that
\begin{equation}\label{eq:43}
\bs p=\frac{s}{\overline{w}}{\bs Q} {\bs
E}_a\bs p=\frac{1}{F_A(2q-1)}\,{\bs Q} {\bs E}_a\bs p,
\end{equation}
where ${\bs E}_a\bs p$ has the form \eqref{eq:41}. The normalizing factor
$\theta$ should be chosen in such a way that \eqref{eq:3} holds. Thus, in
coordinates we have
$$
p_k=\frac{1}{|A|\cdot F_A(2q-1)}\sum\limits_{b\in
A}(1-q)^{H_{kb}}q^{N-H_{kb}},\quad k=0,\ldots,2^N-1,$$ or
\begin{equation}\label{eq:44}
p_k=\frac{1}{|A|}\frac{\sum\limits_{b\in
A}(1-q)^{H_{kb}}q^{N-H_{kb}}}{\sum\limits_{b\in
A}(1-q)^{H_{ab}}q^{N-H_{ab}}}\,,\quad k=0,\ldots,2^N-1,\quad a\in
A.
\end{equation}
The expressions \eqref{eq:44} imply that the distribution $\bs p$
is constant for any fixed $q$ on the $G$-orbits in the set of
indices $A$.

Using the discussed approach, for Example \ref{ex:3:1} and $w=0$ we obtain from \eqref{eq:27} that
\begin{equation}\label{eq:45}
\overline{w}=s \sum\limits_{k=0}^p {\binom{p}{k}}\binom{N-p}{k}(1-q)^{2k}q^{N-2k}.
\end{equation}
%In Example \ref{ex:3:1:1}
%\begin{equation}\label{eq:46}
%\overline{w}=\frac{s}{2^N} \sum_{d=0}^{N} {\binom{N}{d}}(2q-1)^{d}=sq^N.
%\end{equation}
%In Example \ref{ex:3:1:2}
%\begin{equation}\label{eq:47}
%\overline{w}=\frac{s}{2^N}\sum_{d=0}^{N}
%\frac{(N-2d)^2}{N}{\binom{N}{d}}(2q-1)^{d}=s(q^N+Nq^{N-1}(1-q)).
%\end{equation}
In Example \ref{ex:3:1:3} ($N=2n$) we find
\begin{equation}\label{eq:48}
\overline{w}= s \sum\limits_{k=0}^n {\binom{n}{k}}^2(1-q)^{2k}q^{2n-2k}\approx \frac{s}{\sqrt{\pi n}}
\frac{1}{\sqrt{1-(2q-1)^{2}}}\,,\quad\mbox{when}\;\;n\gg
1.
\end{equation}
In Example \ref{ex:3:2} we have
\begin{equation}\label{eq:49}
\overline{w}=\frac{s}{2^{N-1}}\sum_{d=0}^{ \lfloor N/2\rfloor} {\binom{N}{2d}}(2q-1)^{2d}=s(q^N+(1-q)^N).
\end{equation}
Other examples can be treated similarly.
%In Example \ref{ex:3:3} we obtain for $N=8$
%\begin{equation}\label{eq:50}
%\overline{w}(q)=s(q^8+3(1-q)^2q^6+4(1-q)^6q^2).
%\end{equation}
\end{example}

\section{The infinite sequence limit $N\to\infty$}\label{sec:5}
In Corollary \ref{corr:2:1} we
obtained the algebraic equation \eqref{eq:26} of degree at most $N+1$ for
the leading eigenvalue $\overline{w}=\overline{w}(q)$.
The advantage of having a polynomial equation of degree $N+1$ notwithstanding, solving \eqref{eq:26} becomes complicated as
$N\to\infty$.  Moreover, it is well known that at least for some fitness landscapes (including the classical single peaked fitness landscape) the phenomenon of the \textit{error threshold} is observed: there exists a critical mutation rate $q$, after which the quasispecies distribution $\bs p$ becomes uniform. This phenomenon is usually identified with a non-analytical behavior of the limiting eigenvalue $\overline{w}$ when $N\to\infty$, a general idea can be grasped from Fig. \ref{fig:1}b, where it is seen that there exists a corner point on the graph of the function $\overline{w}$.

In this section we propose several steps to rigorously define and analyze this kind of behavior in terms of sequences of orbits $A_n$ that determine our two-valued fitness landscapes. First, we find some bounds for the function
$\overline{w}$ provided $0.5\leq q\leq 1$. Next, we restrict our attention at the special class of sequences $(A_n)_{n=n_0}^\infty$, which we call \textit{admissible} and of the \textit{moderate growth} (here $n_0$ is a sufficiently large natural number). Finally, among all those admissible sequences of the moderate growth we identify the ones that demonstrate some kind of non-uniform convergence for the corresponding sequence of eigenvalues $(\overline{w}^{(n)})_{n=n_0}^\infty$.
\subsection{Lower and upper bounds on $\overline{w}(q)$}\label{sec:5:1}

First we note that for our purposes it is easier to deal with the series \eqref{eq:25}
rather than \eqref{eq:26}. We also make the following substitutions
\begin{equation}\label{eq:51}
w=us\;,\quad \overline{w}=\overline{u}s.
\end{equation}
Then \eqref{eq:25} turns into
\begin{equation}\label{eq:52}
\overline{u}=
\sum_{c=0}^{\infty}\left(\frac{u}{\overline{u}}\right)^c
F_A((2q-1)^{c+1}),
\end{equation}
where the polynomial $F_A(z)$,  defined in \eqref{eq:21}, can be
represented in the form \eqref{eq:24}.

From Example \ref{ex:3:8} we have that $sF_A(2q-1)=\overline{w}(q)$ is the leading eigenvalue if $w=0$.
It was proved in \cite{semenov2014} that $\overline{w}(q)$ increases on the
segment $0.5\leq q\leq 1$. Therefore, on this segment we have
the non-increasing sequence (for any fixed $q$)
\begin{equation}\label{eq:53}
F_A(2q-1)\geq F_A((2q-1)^2)\geq\dots\geq F_A((2q-1)^{c})\geq
F_A((2q-1)^{c+1})\geq\dots >0,
\end{equation}
since $F_A((2q-1)^{c})>0$ according to \eqref{eq:22}. Hence,
\begin{equation}\label{eq:54}
\overline{u}=
\sum_{c=0}^{\infty}\left(\frac{u}{\overline{u}}\right)^c
F_A((2q-1)^{c+1})\leq
F_A(2q-1)\sum_{c=0}^{\infty}\left(\frac{u}{\overline{u}}\right)^c=
\frac{F_A((2q-1))\,\overline{u}}{\overline{u}-u}\,.
\end{equation}
It follows that $ \overline{u}\leq u+ F_A(2q-1)$, or
$$
\overline{w}(q)\leq w+sF_A(2q-1)=:\overline{w}_{up,1}(q).
$$

A second upper bound can be obtained as follows:
\begin{align*}
\overline{u}&=
\sum_{c=0}^{\infty}\left(\frac{u}{\overline{u}}\right)^c
F_A((2q-1)^{c+1})=F_A(2q-1)+
\sum_{c=1}^{\infty}\left(\frac{u}{\overline{u}}\right)^c
F_A((2q-1)^{c+1})\\
&\leq F_A(2q-1)+F_A((2q-1)^2)\sum_{c=1}^{\infty}\left(\frac{u}{\overline{u}}\right)^c=
 F_A(2q-1)+\frac{uF_A((2q-1)^2)}{\overline{u}-u}\,.
\end{align*}
Solving the quadratic inequality we get
$$
\overline{u}\leq\frac{u+F_A(2q-1)+\sqrt{(u+F_A(2q-1))^2-4u(F_A(2q-1)-F_A((2q-1)^2))}}{2}\,,
$$
or,
\begin{equation}\label{eq:55}
\overline{w}(q)\leq\frac{\overline{w}_{up,1}(q)+
\sqrt{\overline{w}^2_{up,1}(q)-4w(sF_A(2q-1)-sF_A((2q-1)^2))}}{2}=:\overline{w}_{up,2}(q).
\end{equation}

\begin{remark} In view of \eqref{eq:53}
$\overline{w}_{up,2}(q)\leq\overline{w}_{up,1}(q)$.

%It should be mentioned that if Conjecture 2.1 is true a better upper bound can
%be presented.
\end{remark}

To obtain a lower bound on $\overline{w}(q)$ we use the approach applied in \cite{semenov2014}.
Since $\overline{w}(q)$ increases on the segment
$0.5\leq q\leq 1$ therefore
\begin{equation}\label{eq:56}
\overline{w}(q)\geq\overline{w}(0.5)=w+\frac{s|A|}{2^N}\,.
\end{equation}

By the definition of \eqref{eq:21}
\begin{align*}
F_A((2q-1)^{c+1})&=\sum_{b\in A}\left(\frac{1-(2q-1)^{c+1}}{2}\right)^{H_{ab}}
\left(\frac{1+(2q-1)^{c+1}}{2}\right)^{N-H_{ab}}\\
&\geq \left(\frac{1+(2q-1)^{c+1}}{2}\right)^{N}\geq
\left(\frac{1+(2q-1)}{2}\right)^{(c+1)N}=q^{(c+1)N}\,,
\end{align*} since
$a\in A$, $H_{aa}=0$ and the function $f(t)=t^{c+1}$ is convex
(downward) on the segment $[0,1]$.

Now from \eqref{eq:52}
\begin{equation}\label{eq:57}
\overline{u}\geq
\sum_{c=0}^{\infty}\left(\frac{u}{\overline{u}}\right)^c
q^{(c+1)N}=\frac{\overline{u}\,q^N}{\overline{u}-uq^N}\,,\quad\mbox{or}\;\;
\overline{u}\geq(u+1)q^N\,,\quad\mbox{or}\;\; \overline{w}\geq
(w+s)q^N.
\end{equation}
Combining \eqref{eq:56} and \eqref{eq:57} yields
\begin{equation}\label{eq:58}
\overline{w}(q)\geq \max\left(w+\frac{s|A|}{2^N},
(w+s)\,q^N\right)=:\overline{w}_{low}(q)\,.
\end{equation}
Thus we have proved
\begin{proposition}
For the leading eigenvalue $\overline{w}(q)$ of \eqref{eq:2} in the case of the two-valued fitness landscape we have
$$
\overline{w}_{low}(q)\leq w(q)\leq \overline{w}_{up,2}(q),\quad 0.5\leq q\leq 1,
$$
where $\overline{w}_{low}(q)$ is given by \eqref{eq:58}, and $\overline{w}_{up,2}(q)$ is given by \eqref{eq:55}.
\end{proposition}

A numerical example with the obtained bounds is given in Figure \ref{fig:2}.
\begin{figure}[!ht]
\centering
\includegraphics[width=0.95\textwidth]{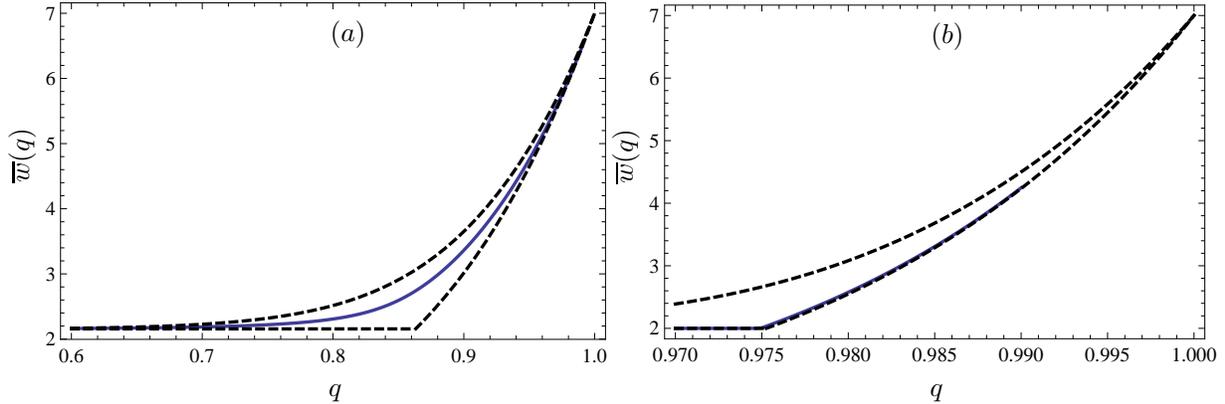}
\caption{The lower and upper bounds on the leading eigenvalue $\overline{w}$ in the case of the quaternion landscape (Example \ref{ex:3:3}), $(a)$ $N=8$, $(b)$ $N=50$.}\label{fig:2}
\end{figure}
\subsection{Admissible sequences of orbits}
To make a progress in analyzing the limit behavior of our system when $N\to\infty$ we introduce in this subsection two definitions in terms of which this behavior will be described.

From the previous subsection, we see that the curve $\overline{w}=\overline{w}_{low}(q)$ has a corner point on $[0.5,1]$, which we denote $q_\ast$:
\begin{equation}\label{eq:59}
q_*=q_*^{(N)}=\sqrt[N]{\frac{w+s|A|2^{-N}}{w+s}}=\sqrt[N]{\frac{u+|A|2^{-N}}{u+1}}=\sqrt[N]{\frac{\overline{w}(0.5)}{\overline{w}(1)}}\,.
\end{equation}
The function $\overline{w}_{low}(q)$ is constant for $0.5\leq
q\leq q_*$ and increases for $q_*<q\leq 1$ (see Figure \ref{fig:2}). It
was shown in \cite{semenov2014} that for the single peak landscapes ($|A|=1$) the lower bound $\overline{w}_{low}(q)$ provides a
close approximation for $\overline{w}(q)$ for sufficiently large $N$. Our goal is to generalize these results on the case of the two-valued fitness landscapes.

From this point on we shall use $n$ as the index, which tends to infinity. In most cases it actually coincides with the sequence length $N$, albeit not always, hence the choice of notation.

One of the main underlying questions concerning the quasispecies model and especially its infinite sequence limit, is how actually the fitness landscape is scaled when $N\to\infty$. In most works in literature a continuous limit is used, which basically narrows the pull of the allowed fitness landscapes to the ones which have, given this continuous limit, a limit fitness function, which must be also continuous (e.g., \cite{Baake2007,saakian2006ese}). Here we take a different approach by specifying sequences of orbits $(A_n)_{n=n_0}^\infty$, on which the fitness landscape is defined. The sequences that are of interest to us will be called \textit{admissible}.

Suppose that for any $n\geq n_0$  a sequence of $G_n$-orbits
$A_n\in X_n$ is given, where $G_n\leqslant{\rm Iso}\, (X_n)$. When
$n\to \infty$ the group ${\rm Iso}\, (X_n)$ will be always viewed
as a subgroup of ${\rm Iso}\, (X_{n+1})$. More precisely, let
$g\in {\rm Iso}\, (X_n)$ be a fixed isometry and let $a\in
X_{n+1}$ be represented as $a=a_{n}+\alpha_{n} 2^{n}$ where
$a_{n},\alpha_{n}\in X_n$. Then $g$, viewed as an element of ${\rm
Iso}\, (X_{n+1})$, maps $a\in X_{n+1}$ to
$g(a):=g(a_{n})+g(\alpha_{n}) 2^{n}$. In other words, ${\rm Iso}\,
(X_n)$ as a subgroup of ${\rm Iso}\, (X_{n+1})$ is acting on the
``upper'' hyperface $V_n\times\{1\}$ of the cube
$V_{n+1}=\{0,1\}^{n+1}=V_n\times\{0,1\}$ in the same way as it
acts on the ``lower'' hyperface $V_n\times\{0\}\cong V_n$. Thus,
we have the ascending chain
$$
{\rm Iso}\, (X_{n_0})<\ldots< {\rm Iso}\, (X_{n})<{\rm Iso}\,
(X_{n+1})<\ldots
$$
and the corresponding ascending chain of symmetric subgroups
$$
S_{n_0}<\ldots< S_{n}<S_{n+1}<\ldots\,.
$$

For a fixed $w\geq 0$ consider a sequence of landscapes $(\bs w^{(n)})_{n\geq n_0}$ such that
$w_k^{(n)}=w+s$ if $k\in A_n$ and $w_k^{(n)}=w$ otherwise. The
sequence $(A_n)_{n=n_0}^\infty$ and the parameters $w$, $s$, and $u=w/s$ define the
corresponding family of leading eigenvalues
$\overline{w}^{(n)}=\overline{w}^{(n)}(q)$, which are solutions of
\eqref{eq:2}, and the family $\overline{u}^{(n)}=\overline{u}^{(n)}(q)$,
such that $\overline{u}^{(n)}=\overline{w}^{(n)}/s$.

In \cite{semenov2014}  it was proved that for any $n\geq n_0$  the
function $\overline{u}^{(n)}(q)$ has the following properties:
\begin{enumerate}
\item The function $\overline{u}^{(n)}(q)$ increases
on the segment $[0.5,1]$ and is convex (downward) there.

\item $\overline{u}^{(n)}(0.5)=u+\displaystyle{\frac{|A_n|}{2^n}},\quad \overline{u}^{(n)}(1)=u+1$.
\end{enumerate}

\begin{definition}\label{d:5:1} A sequence
$(A_n)_{n=n_0}^\infty$ of $G_n$-orbits is called admissible if the
corresponding sequence of values of polynomials $F_{A_n}(2q-1)$ in \eqref{eq:23} is non-increasing for each $q\in [0.5,1]$:
\begin{equation}\label{eq:61}
F_{A_n}(2q-1)=\sum_{d=0}^n f^{(n)}_{d}\,(1-q)^dq^{n-d}\geq
\sum_{d=0}^{n+1}
f^{(n+1)}_{d}\,(1-q)^dq^{n+1-d}=F_{A_{n+1}}(2q-1)\,,\;n\geq
n_0.
\end{equation}
\end{definition}

\begin{definition} A sequence
$(A_n)_{n=n_0}^\infty$ of $G_n$-orbits is called a sequence of the moderate
growth if
\begin{equation}\label{eq:62}
\lim_{n\to\infty}\frac{|A_n|}{2^n}=0\,,\quad\mbox{or}\quad |A_n|=o(2^n),\quad n\to\infty\,.
\end{equation}
\end{definition}

To show that our definitions make sense we state
\begin{proposition}\label{pr:ad}In all the examples of Section \ref{sec:3} the corresponding sequences of orbits are admissible and of the moderate growth.
\end{proposition}
\begin{proof}See Appendix \ref{ap:1}.
\end{proof}

Consider a sequence $(A_n)_{n=n_0}^\infty$ of $G_n$-orbits. Our next aim is to
investigate what happens with the corresponding family
$(\overline{u}^{(n)})_{n=n_0}^\infty$ as $n\to\infty$.

\begin{proposition}\label{pr:5:1} If $(A_n)_{n=n_0}^\infty$ is an
admissible sequence of $G_n$-orbits then for each fixed
$q\in[0.5,1]$ the sequence $(\overline{u}^{(n)}(q))_{n=n_0}^\infty$ is a
non-increasing sequence as $n\to\infty$. If, additionally,
$(A_n)_{n=n_0}^\infty$ is a sequence of the moderate growth then
$\lim\limits_{n\to\infty}\overline{u}^{(n)}(0.5)=u$ and
$\lim\limits_{n\to\infty}\overline{u}^{(n)}(1)=u+1$.
\end{proposition}

\begin{proof} The second assertion follows directly from Property 2
of $\overline{u}^{(n)}(q)$ above. Let us proof the first one. The
equation \eqref{eq:52} for $u\ne 0$ can be rewritten in the form
\begin{equation}\label{eq:69}
u=\frac{u}{\overline{u}^{(n)}(q)}\overline{u}^{(n)}(q)=
\sum_{c=0}^{\infty}\left(\frac{u}{\overline{u}^{(n)}(q)}\right)^{c+1}
F_{A_n}((2q-1)^{c+1})=\sum_{m=1}^{\infty}\left(\frac{u}{\overline{u}^{(n)}(q)}\right)^{m}
F_{A_n}((2q-1)^m)\,.
\end{equation}

It follows from Definition \ref{d:5:1} that at each fixed point
$q\in[0.5,1]$ the sequence of positive coefficients
$(F_{A_n}((2q-1)^{c+1}))_{n\geq n_0}$ is non-increasing for any $c+1\in
\N$. But the left-hand side $u$ of \eqref{eq:69} is constant. This
implies that $(\overline{u}^{(n)}(q))_{n\geq n_0}$ must be a non-increasing
sequence for each $q\in[0.5,1]$.
\end{proof}

Hence we can conclude that the curve
$\overline{u}=\overline{u}^{(n+1)}(q)$ always passes {\it under}
the curve $\overline{u}=\overline{u}^{(n)}(q)$ in the rectangle
$\{0.5\leq q\leq 1\,,\;u\leq\overline{u}\leq u+1\}$ if $(A_n)_{n=n_0}^\infty$
is an admissible sequence of $G_n$-orbits, see Figure \ref{fig:3}.

Proposition \ref{pr:5:1} and Property 1 of $\overline{u}^{(n)}(q)$ yield
\begin{corollary} If $(A_n)_{n=n_0}^\infty$ is an admissible sequence
of $G_n$-orbits of the moderate
growth then for any fixed $\varepsilon\in(0,1]$ there exists
$N_0\in \N$ such that for any $n\geq N_0$ the curve
$\overline{u}=\overline{u}^{(n)}(q)$ intersects the line
$\overline{u}=u+\varepsilon$ at a unique point
$q^{(n)}(\varepsilon,u)\in(0.5,1]$.
\end{corollary}

Note that by virtue of \eqref{eq:52}, \eqref{eq:69}, and \eqref{eq:26}
the value $q^{(n)}(\varepsilon,u)$ from the previous corollary can be found from one of the following
equations
\begin{equation}\label{eq:70}
u+\varepsilon=
\sum_{c=0}^{\infty}\left(\frac{u}{u+\varepsilon}\right)^{c}
F_{A_n}((2q-1)^{c+1})\,,
\end{equation}
%or
%\begin{equation}\label{eq:71}
%u= \sum_{m=1}^{\infty}\left(\frac{u}{u+\varepsilon}\right)^{m}
%F_{A_n}((2q-1)^{m}),\,\quad u\ne 0\,,
%\end{equation}
or,
\begin{equation}\label{eq:728}
\sum_{d=0}^N \frac{h_d(2q-1)^d}{u+\varepsilon-u(2q-1)^d}=1.
\end{equation}

Another almost immediate result is given in the following
\begin{proposition} If $(A_n)_{n\geq N_0}$ is an
admissible sequence of $G_n$-orbits of the moderate growth then for
fixed $(\varepsilon,u)$ the sequence
$(q^{(n)}(\varepsilon,u))_{n\geq N_0}$ is non-decreasing as
$n\to\infty$ and the inequality
\begin{equation}\label{eq:72}
q^{(n)}(\varepsilon,0)\leq q^{(n)}(\varepsilon,u)\leq
\sqrt[n]{\frac{u+\varepsilon}{u+1}}\leq
1-\frac{1-\varepsilon}{n(u+1)}\,
\end{equation}
holds.
\end{proposition}

\begin{proof} The upper bound (see Section \ref{sec:5:1})
$\overline{u}=u^{(n)}_{up,1}(q)= u+ F_{A_n}(2q-1)$ gives rise to
the lower bound in \eqref{eq:72} since the equation $u+\varepsilon=u+
F_{A_n}(2q-1)$ is equivalent to \eqref{eq:70} when $u=0$. The lower bound
\eqref{eq:57} $\overline{u}=(u+1)q^n$ provides the upper bound in \eqref{eq:72}.

Since $\overline{u}=(u+1)q^n$ is convex downward if $q\in[0.5,1]$
and $\overline{u}=u+1-n(u+1)(1-q)$ is the equation of the tangent
at $q=1$ to the curve $\overline{u}=(u+1)q^n$ then we get the last
inequality in \eqref{eq:72}. Note that the curve
$\overline{u}=\overline{u}^{(n)}(q)$ has the same tangent at $q=1$
(see, for instance, \cite{semenov2014}).
\end{proof}
The obtained results are illustrated in Figure~\ref{fig:3}.
\begin{figure}[!ht]
\centering
\includegraphics[width=0.5\textwidth]{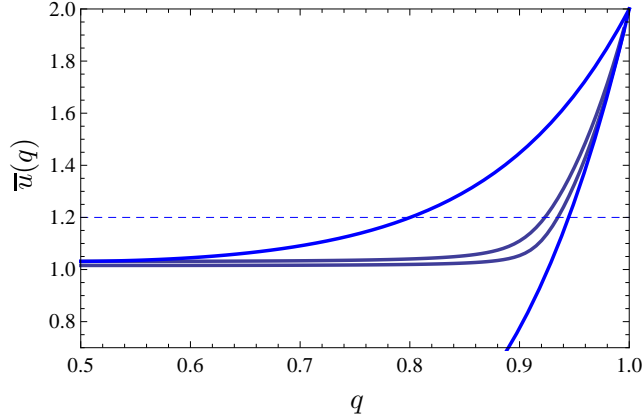}
\caption{The curves in the coordinates $q,\overline{u}$ defined by (from top to bottom): $\overline{u}=u+F_{A_n}(2q-1),\,\overline{u}=\overline{u}^n(q),\,\overline{u}=\overline{u}^{n+1}(q),\,\overline{u}=(u+1)q^{n+1}$. The points of intersections of these curves with the dotted line $\overline{u}=u+\varepsilon$ define the values $q^{(n)}(\varepsilon,0),\,q^{(n)}(\varepsilon,u),\,q^{(n+1)}(\varepsilon,u),\,1-\frac{1-\varepsilon}{(n+1)(u+1)}$ respectively, see also \eqref{eq:72}.}\label{fig:3}
\end{figure}

\subsection{Threshold-like behavior}
In this subsection we define rigorously what we call the threshold-like behavior and provide sufficient conditions for the sequences of admissible orbits to possess this kind of behavior. The main conclusion, which can be stated in a form of a conjecture, emphasizes the role of geometry for the threshold-like behavior to occur. Loosely speaking, if the admissible sequence of orbits ``looks like a point'' asymptotically, i.e., basically indistinguishable from the single peaked landscape in the infinite length limit, then the threshold-like behavior is observed. We conjecture, as numerical experiments show, that the opposite is true: If asymptotically the admissible sequence of orbits is different from a point, then there exits no threshold-like behavior.

Let us introduce the notation
\begin{equation}\label{eq:73}
q^{(n)}_*(\varepsilon,u)=\sqrt[n]{\frac{u+\varepsilon}{u+1}}\,,
\end{equation}
from where
\begin{equation}\label{eq:74}
\lim_{n\to\infty}n(1-q^{(n)}_*(\varepsilon,u))=\log\frac{u+1}{u+\varepsilon}\,.
\end{equation}
It follows that for a fixed $u>0$
\begin{equation}\label{eq:75}
\lim_{\varepsilon\downarrow 0}\lim_{n\to\infty}n(1-q^{(n)}_*(\varepsilon,u))=\log\frac{u+1}{u}\,.
\end{equation}

It is known (e.g., \cite{semenov2014}) that for the single peaked landscape the
curve $\overline{u}=\overline{u}^{(n)}(q)$ passes very close to the lower bound $\overline{u}=\max\{u,(u+1)q^n\}$
in such a way that
$$q^{(n)}(\varepsilon,u)=q^{(n)}_*(\varepsilon,u)-o\left(\frac{1}{n}\right)=
\sqrt[n]{\frac{u+\varepsilon}{u+1}}-o\left(\frac{1}{n}\right)$$ as
$n\to\infty$ (from \eqref{eq:72} we have the inequality
$q^{(n)}(\varepsilon,u)\leq q^{(n)}_*(\varepsilon,u)$).

%
%For an admissible sequence of orbits $(A_n)_{n\geq n_0}$ consider an
%increasing (in view of Proposition \ref{pr:5:1}) chain of compact convex
%sets
%$$C_n=\{(q,\overline{u})\,|\,0.5\leq q\leq 1,\; \overline{u}^{(n)}(q)\leq \overline{u}\leq u+1\}\,$$
% and the union $C=\bigcup\limits_{n} C_n$. The lower boundary of
% $C_n$ is the curve $\overline{u}=\overline{u}^{(n)}(q)$.

Our next aim  is to investigate what happens with the curve $\overline{u}=\overline{u}^{(n)}(q)$ as $n\to\infty$. It is more conveniently done in coordinates $x$,
 $L$, defined by
\begin{equation}\label{eq:76}q=1-\frac{x}{n}\,,\quad 0\leq x\leq \frac{n}{2}\,,\quad \overline{u}=(u+1)e^{-L}\,,\quad 0\leq L\leq
 \log\frac{u+1}{u}\,.
 \end{equation}
We will assume that $u>0$ in \eqref{eq:76}. Hence the curve
$\overline{u}=\overline{u}^{(n)}(q)$ transforms into the curve
\begin{equation}\label{eq:77}
L_n(x)=\log(u+1)-\log\overline{u}^{(n)}\left(
1-\frac{x}{n}\right).
\end{equation}
Note that $L_n(0)=0$ since ${u}^{(n)}(1)=u+1$ for any $n$.
\begin{definition}\label{def:5:3} We say that an admissible sequence $(A_n)_{n\geq n_0}$
of the moderate growth, or, equivalently, the family
$(\overline{u}^{(n)})_{n\geq n_0}$ possesses the threshold-like behavior on the segment $[0.5,1]$ if for each
fixed $x\geq 0$ and the corresponding functions $L_n(x)$ it is true that
\begin{equation}\label{eq:78}
\lim_{n\to\infty} L_n(x)=L(x)=\begin{cases}x,&0\leq x<\log \frac{u+1}{u},\\
\log\frac{u+1}{u}\,,&x\geq \log\frac{u+1}{u}\,.
\end{cases}
\end{equation}
\end{definition}
The definition above is illustrated in Fig.~\ref{fig:4}.
\begin{figure}[!ht]
\centering
\includegraphics[width=0.55\textwidth]{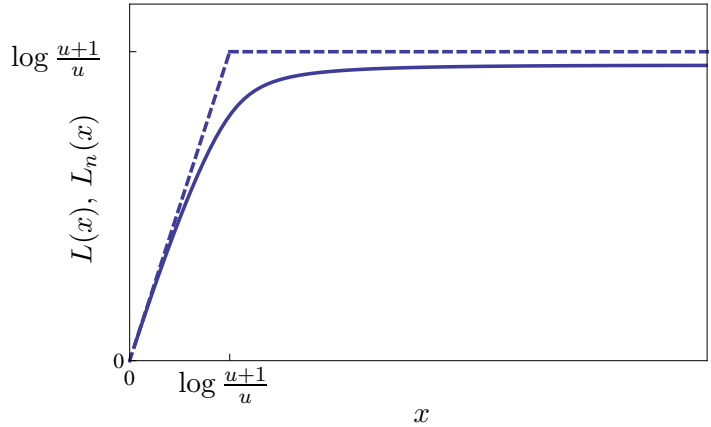}
\caption{The limit function $L(x)$ in Definition \ref{def:5:3} of the threshold-like behavior}\label{fig:4}
\end{figure}

%\begin{picture}(400,130)
%
%\put(130,40){{\vector(1,0){200}}} \put(140,30){{\vector(0,1){90}}}
%\put(140,40){{\circle*{2}}} \put(140,90){{\circle*{2}}}
%\put(190,40){{\circle*{2}}} \put(120,20){0}
%
%\put(330,25){$x$} \put(125,110){$L$}
%\put(105,85){$\log\frac{u+1}{u}$} \put(180,25){$\log\frac{u+1}{u}$}
%\put(200,70){$L=L_n(x)$} \put(180,100){$L=L(x)$}
%
%\qbezier(220,89)(208,89)(191,87)\qbezier(140,40)(182,85)(191,87)
%
%\linethickness{0.7pt}{\put(190,90){{\line(1,0){120}}}
%\put(140,40){{\line(1,1){50}}}  }
% \put(120,5){Figure$^{**}$. Error threshold in coordinates $x$, $L$}
%\end{picture}

If the threshold-like behavior is present in the two-valued fitness landscape, then the following formula provides an approximation for the
threshold mutation rate $q_*^{(n)}(u)$, $n\gg 1$:
\begin{equation}\label{eq:79}
q^{(n)}_*(u)\approx
1-\frac{1}{n}\log\frac{u+1}{u}=1-\frac{1}{n}\log\frac{w+s}{w}\,,
\end{equation}
which, of course, coincides with the classical estimate for the error threshold for the single peaked landscape \cite{eigen1988mqs,semenov2014}.

If the sequence of continuous functions
$(\overline{u}^{(n)})_{n\geq n_0}$ has the threshold-like behavior then it converges not uniformly on $[0.5,1]$, as
$n\to\infty,$ to the discontinuous function $\psi(q)$ such that
$\psi(q)=u$ if $0.5\leq q<1$ and $\psi(1)=u+1$.

The following theorems and corollaries provide  sufficient
conditions under which an admissible sequence of orbits
of the moderate growth shows the threshold-like behavior.

\begin{theorem}\label{th:5:1}In the above notation suppose that for
$n\geq n_0$ an admissible sequence of $G_n$-orbits $A_n\subset
X_n$ \emph{(}${G}_n\leqslant{\rm Iso}\, (X_n)$\emph{)} of the
moderate growth is given and $u>0$. Suppose also that for $n\geq
n_0$  the inequality
\begin{equation}\label{eq:80}
F_{A_n}(2q-1)\leq (2q-1)^{n/2}+M_n\,,\quad
0.5\leq q\leq 1 \,,
\end{equation}
is satisfied for some constants $M_n$ such that $\lim\limits_{n\to \infty}M_n=0$. Then the sequence $(A_n)_{n\geq n_0}$ shows the threshold-like behavior on the segment
$[0.5,1]$.
\end{theorem}

\begin{proof}  In view of equation \eqref{eq:52}, in coordinates $x$,
$L$:
$$
(u+1)e^{-L_n(x)}=
\sum_{c=0}^{\infty}\left(\frac{u}{(u+1)e^{-L_n(x)}}\right)^c
F_{A_n}\left(\left(1-\frac{2x}{n}\right)^{c+1}\right),
$$
therefore (putting $m=c+1$)
\begin{equation}\label{eq:81}
u= \sum_{m=1}^{\infty}\left(\frac{u}{u+1}\right)^m
e^{mL_n(x)}F_{A_n}\left(\left(1-\frac{2x}{n}\right)^{m}\right).
\end{equation}

The lower bound \eqref{eq:57}, $\overline{u}\geq(u+1)q^n$, implies for
$x\in\left[0,\log\frac{u+1}{u}\right)$, $n\gg 1$,
\begin{equation}\label{eq:82}
L_n(x)=\log\frac{u+1}{\overline{u}^{(n)}(1-\frac{x}{n})}\leq
-n\log\left(1-\frac{x}{n}\right).
\end{equation}
 Consequently, we have on
$\left[0,\log\frac{u+1}{u}\right)$
\begin{equation}\label{eq:83}
\limsup_{n\to\infty}
L_n(x)\leq -\lim_{n\to\infty}n\log\left(1-\frac{x}{n}\right)=
x.
\end{equation}

On the other hand, the   function $F_{A_n}(2q-1)$, as the leading
eigenvalue for $u=0$ (see \eqref{eq:27}), is increasing on the segment
$[0.5, 1]$. In view of the inequality
$1-t\leq e^{-t}$  we have
$$F_{A_n}\left(\left(1-\frac{2x}{n}\right)^{m}\right)\leq
F_{A_n}\left(e^{-2mx/n}\right)\,.$$
Make the substitution $2q-1=e^{-2x/n}$  into \eqref{eq:80}, where $0\leq
x<+\infty$. Then the following inequality
$$
e^{mL_n(x)}F_{A_n}\left(\left(1-\frac{2x}{n}\right)^{m}\right)\leq
e^{mL_n(x)}F_{A_n}\left(e^{-2mx/n}\right)\leq
e^{m(L_n(x)-x)}+M_ne^{mL_n(x)}\,
$$
holds. Hence, \eqref{eq:81} yields
$$
u\leq\sum_{m=1}^{\infty}\left(\frac{u}{u+1}\right)^m
e^{m(L_n(x)-x)}+M_n\sum_{m=1}^{\infty}\left(\frac{ue^{L_n(x)}}{u+1}\right)^m
\,.
$$
In view of \eqref{eq:82} both progressions in the right-hand side converge
for $x\in \left[0,\log\frac{u+1}{u}\right)$ and $n\gg 1$.

The simplification provides the inequality
$$
1\leq e^{L_n(x)-x}+
M_n\frac{e^{L_n(x)}(u+1-ue^{L_n(x)-x})}{(u+1)(u+1-ue^{L_n(x)})}\,.
$$
Since $M_n\to 0$ and the inequality \eqref{eq:82} holds we get finally $
\liminf\limits_{n\to\infty} e^{L_n(x)-x}\geq 1\,$, or,
\begin{equation}\label{eq:84}
\liminf_{n\to\infty} L_n(x)\geq x\,.
\end{equation}

It follows from \eqref{eq:83}, \eqref{eq:84} that $\lim\limits_{n\to\infty}L_n(x)=x$
if $0\leq x< \log\frac{u+1}{u}$. The increasing functions
$L_n(x)$ cannot exceed the value
$\log\frac{u+1}{u}$. Thus, the threshold-like behavior is observed.
\end{proof}

\begin{corollary}\label{corr:5:2} The sequence of constant single peaked landscapes
$A_n\equiv \{a\},\,n\geq n_0$ possesses the threshold-like behavior.
\end{corollary}

\begin{proof} Condition \eqref{eq:80} of Theorem \ref{th:5:1} reads as follows: the inequality
\begin{equation}\label{eq:85}q^n\leq (2q-1)^{n/2}+M_n,\quad
0.5\leq q\leq 1,
\end{equation} holds for some constants $M_n$ such
that $\lim\limits_{n\to \infty}M_n=0$.

Let us show that the constants $$M_n=\frac{1}{n}\left(1-\frac{1}{n}\right)^{n-1}\leq\frac{1}{e(n-1)}$$ fit. Consider
the auxiliary function $\varphi_n(q)=q^n- (2q-1)^{n/2}$. Its
maximum $\mu_n$ on the segment $[0.5,1]$ is reached at some point
$q_n<1$. This point is a root of the equation
$$
\varphi'_n(q)=nq^{n-1}-n (2q-1)^{n/2-1}=0\;, \quad
\mbox{or}\;\;(2q-1)^{n/2}=(2q-1)q^{n-1}\,.
$$
Hence, by the definition of $\varphi_n(q)$, we get
$$\mu_n=\varphi_n(q_n)=q_n^n-(2q_n-1)q_n^{n-1}=q_n^{n-1}(1-q_n)\,.$$
The function $M(t)=t^{n-1}(1-t)$ achieves its maximum on [0,1], which is
equal to $M_n=\frac{1}{n}\left(1-\frac{1}{n}\right)^{n-1}$, at
$t_n=1-\frac{1}{n}$. Hence, $\mu_n\leq M_n$ and the conditions of
Theorem \ref{th:5:1} hold.
\end{proof}

Theorem \ref{th:5:1} together with Corollary \ref{corr:5:2} are the key results as the following
theorem shows. We are convinced that the reason for the error
threshold effect is  geometric. More precisely, in view of \eqref{eq:23} the polynomial $F_{A_n}(2q-1)$ can be always represented in
the form
\begin{equation}\label{eq:86}F_{A_n}(2q-1)=q^n+\sum_{k=1}^n f_k^{(n)}\,(1-q)^{k}q^{n-k},
\end{equation}
where $f_k^{(n)}=\#\{b\in A_n\,|\,H_{ab}=k\}$. Thus, this
polynomial can be viewed as a kind of the spherical growth
function of the orbit $A_n$ with respect to an arbitrary fixed
point $a\in A_n$.

\begin{theorem}\label{th:5:2} In the above notation assume that for any
$n\geq n_0$ an admissible sequence of $G_n$-orbits $A_n\subset
X_n$ \emph{(}$G_n\leqslant{\rm Iso}\,( X_n)$\emph{)}, $n\geq n_0$
of the moderate growth is given and $u>0$. If either
\begin{equation}\label{eq:87}
\lim_{n\to\infty}\max_{q\in [0.5,1]}\sum_{k=1}^{\lfloor n/2\rfloor}
f_k^{(n)}\,(1-q)^{k}q^{n-k}=0\,,
\end{equation}
or
\begin{equation}\label{eq:88}
\lim_{n\to\infty}\sum_{k=1}^{\lfloor n/2\rfloor}
f_k^{(n)}\,\left(\frac{k}{n}\right)^{k}\left(1-\frac{k}{n}\right)^{n-k}=0\,,
\end{equation}
then the sequence $(A_n)_{n\geq n_0}$ possesses the threshold-like behavior.
\end{theorem}

\begin{proof}The polynomial $P_k^{(n)}(q)=(1-q)^kq^{n-k}$
decreases on the segment [0.5,1] if $n<2k\leq 2n$ and achieves its
maximal value $2^{-n}$ at $q=0.5$. If $1\leq k\leq \lfloor n/2\rfloor$ then the
maximal value of $P_k^{(n)}(q)=(1-q)^kq^{n-k}$  on [0.5,1] is
achieved at the point $q_k^{(n)}=1-\frac{k}{n}$ and is equal to
$\left(\frac{k}{n}\right)^{k}\left(1-\frac{k}{n}\right)^{n-k}$.

Denote $F_n=\max\limits_{q\in
[0.5,1]}\sum_{k=1}^{\lfloor n/2\rfloor}f_k^{(n)}(1-q)^{k}q^{n-k}$. Then
Corollary \ref{corr:5:2}, \eqref{eq:85} and \eqref{eq:86} together yield
\begin{equation}\label{eq:89}
\begin{split}
F_{A_n}(2q-1)&=q^n+\sum_{k=1}^n f_k^{(n)}(1-q)^{k}q^{n-k} \leq  q^n+F_n+ \sum_{k=\lfloor n/2\rfloor+1}^n \frac{f_k^{(n)}}{2^n}\\
&\leq (2q-1)^{n/2}+\frac{1}{n}\left(1-\frac{1}{n}\right)^{n-1}+F_n+\frac{|A_n|}{2^n}=(2q-1)^{n/2}+o(1)\,,\quad n\to \infty.
\end{split}
\end{equation}

On the other hand, if the equality \eqref{eq:88} holds we can substitute
$\sum_{k=1}^{\lfloor n/2\rfloor}
f_k^{(n)}\,\left(\frac{k}{n}\right)^{k}\left(1-\frac{k}{n}\right)^{n-k}$
for $F_n$ since $F_n\leq \sum_{k=1}^{\lfloor n/2\rfloor}
f_k^{(n)}\left(\frac{k}{n}\right)^{k}\left(1-\frac{k}{n}\right)^{n-k}$.
 Hence, Theorem \ref{th:5:1} implies that  the sequence
$(A_n)_{n\geq n_0}$ shows the threshold-like behavior.
\end{proof}

\begin{corollary}The following sequences of orbits possess the threshold-like behavior:
\begin{itemize}
\item[$(i)$]All the constant sequences $A_n\equiv A$ \emph{(}see,
for instance, Example \ref{ex:3:1:3} of the quaternion landscape\emph{)};

\item[$(ii)$]All the antipodal sequences $A_n=
\{a,a^*\}\subset X_n$ \emph{(}see Example \ref{ex:3:2}\emph{)};

\item[$(iii)$] All the permutation invariant sequences
$A_{p,n}$ where $ A_{p,n}=\{a\in X\,|\,H_a=p\}\;$, $
p=0,1,\dots,n\,$, and $p$ does not depend on $n\geq p$ \emph{(}see
Examples \ref{ex:3:1} and \ref{ex:3:1:2}\emph{)}.
\end{itemize}
\end{corollary}
\begin{proof} $(i)$ Let $A\subset X_{n_0}$. Since the orbit is
fixed then for $n\geq n_0$ all the coefficients $f_k^{(n)}\equiv
f_k^{(n_0)}$ and $f_k^{(n)}\equiv 0$ when $k>n_0$. It follows that
the assumption \eqref{eq:88} of Theorem \ref{th:5:2} that
$$
\lim_{n\to\infty}\sum_{k=1}^{\lfloor n_0/2\rfloor}
f_k^{(n_0)}\,\left(\frac{k}{n}\right)^{k}\left(1-\frac{k}{n}\right)^{n-k}=0\,
$$
holds since $k\leq \lfloor n_0/2\rfloor$ and, consequently,
$\left(\frac{k}{n}\right)^{k}\to 0$ as $n\to\infty$ provided
$\left(1-\frac{k}{n}\right)^{n-k}\leq 1$.

$(ii)$ In this case $F_{A_n}(q)=q^n+(1-q)^n$.
Then $f_k^{(n)}=0$, $k=1,\dots, \lfloor n/2\rfloor$. Then both assumptions
\eqref{eq:87}, \eqref{eq:88} hold.

$(iii)$ We may suppose that $n\ge 2p=n_0$. In view of
\eqref{eq:29}
$$ F_{A_{p,n}}(q)=\sum\limits_{k=0}^p
{\binom{p}{k}}{\binom{n-p}{k}}(1-q)^{2k}q^{n-2k}\,.
$$
Hence, since $p$ is fixed, ${\binom{p}{k}}<2^p$, $1-\frac{2k}{n}\leq
1$:
 $$\sum\limits_{k=1}^p {\binom{p}{k}}{\binom{n-p}{k}}\left(\frac{k}{n}\right)^{2k}\left(1-\frac{2k}{n}\right)^{n-2k}
\leq (2p)^{2p}\sum\limits_{k=1}^p \frac{1}{n^{2k}}\,{\binom{n-p}{k}}\leq (2p)^{2p}\sum\limits_{k=1}^p \frac{n^k}{k!\,n^{2k}}=o(1)$$
as $n\to\infty$. Consequently, the condition \eqref{eq:88} is satisfied.
\end{proof}

\begin{remark}\label{remark:5}In contrast, if $A_{2n}=A_{n,2n}$ is the sequence of
the fitness landscapes in Example \ref{ex:3:1:3} then the numerical
calculations (Fig. \ref{fig:5}) show that this sequence does not demonstrate the threshold-like behavior. The
approximate formula \eqref{eq:48} provides a lower bound
$$\max_{q\in [0.5,1]}\sum_{k=1}^{n}
{\binom{n}{k}}^2\,(1-q)^{2k}q^{2n-2k}\geq 0.183\approx
\frac{4r-1}{8r\sqrt{\pi r}}\,,\quad
r=-\frac{3}{4}W\left(-\frac{1}{3\sqrt[3]{\pi}}\right)\approx
1.7423,$$ where
$W(z)$ is (a branch of) the Lambert $W$ function
($W(z)e^{W(z)}=z$). Hence, the sufficient conditions for the threshold-like behavior are not satisfied in this case.
\begin{figure}[!t]
\centering
\includegraphics[width=0.95\textwidth]{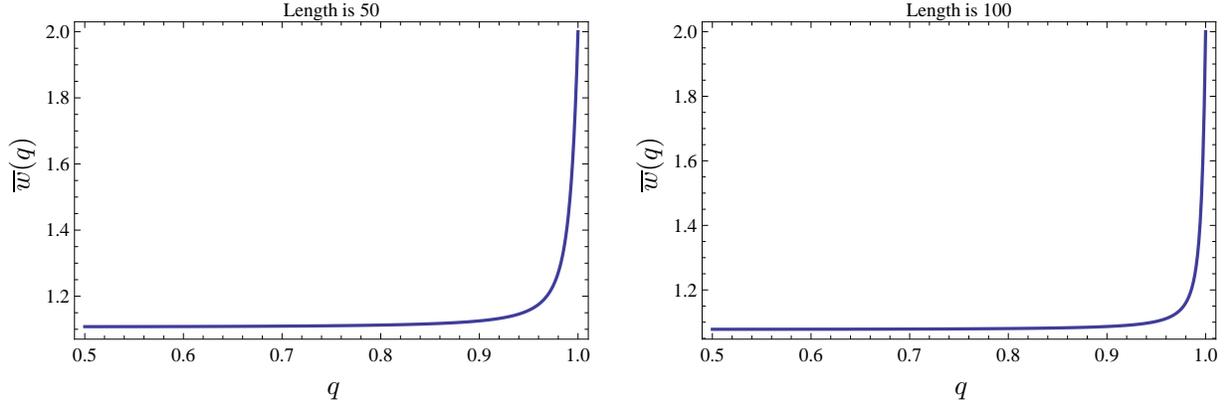}
\caption{The leading eigenvalue $\overline{w}(q)$ versus the mutation rate $q$ in the case of the fitness landscape in Example \ref{ex:3:1:3}. The sequence length $2n=50$ in the left panel and $2n=100$ in the right panel. Note the absence of the threshold-like behavior}\label{fig:5}
\end{figure}
\end{remark}

A natural question to ask is whether the given sufficient conditions are also necessary for the threshold-like behavior. While at this point we do not have a full answer for this question, we can present a sufficient condition for the absence of the threshold like behavior of the sequence $(\overline{u}^{(n)})_{n\geq n_0}$ as
$n\to\infty$. This sufficient condition shows in a way that the condition \eqref{eq:87} is ``almost''
necessary for the error threshold.

\begin{proposition} Suppose that there exist
constants $\varepsilon>0$, $x>0$ such that for all $n\gg n_0$ the
inequality
\begin{equation}\label{eq:90}
F_{A_n}(2q_n-1)\geq (u+1)(q_n^n+2\varepsilon)\,,\quad
q_n=1-\frac{x}{n}\,,
\end{equation}
holds for sufficiently small $u>0$. Then the sequence $(A_n)_{n\geq n_0}$
possesses no threshold-like behavior.
\end{proposition}

\begin{proof} We can assume that $x<\log\frac{u+1}{u}$
for sufficiently small $u>0$ and $q_n>0.5$ for sufficiently large
$n$. In view of \eqref{eq:52}
$$
\overline{u}^{(n)}(q_n)=F_{A_n}(2q_n-1)+\sum_{c=1}^{\infty}\left(
\frac{u+1}{\overline{u}^{(n)}(q_n)}\right)^c
F_{A_n}\left((2q_n-1)^{c+1}\right)\geq F_{A_n}(2q_n-1)\geq
(u+1)(q_n^n+2\varepsilon)\,.
$$
Hence,
$$
L_n(x)=\log(u+1)-\log\overline{u}^{(n)}(q_n)\leq-\log
(q_n^n+2\varepsilon)=-\log\left(\left(1-\frac{x}{n}\right)^n+2\varepsilon\right)<-\log(e^{-x}+\varepsilon).
$$
for $n\gg n_0$. Consequently, $\limsup\limits_{n\to\infty}
L_n(x)\leq -\log(e^{-x}+\varepsilon)<x$.
\end{proof}

Note that in Remark \ref{remark:5} we can take $x=r\approx 1.7423$, $u\leq 0.1$,
$\varepsilon=0.01$, $n\ge 4$.

\section{General construction for the Eigen evolutionary problem}\label{sec:6}
The classical Eigen quasispecies model uses as the underlying geometry the $N$-dimensional hypercube. The distances between the vertices of this hypercube are measured by the number of edges connecting them. While this geometry has a transparent biological interpretation in terms of sequences composed of zeroes and ones, which can be identifies with, e.g., purine and pyrimidine, we feel that it is a natural generalization to consider an arbitrary isometry group acting on an abstract metric space to move to a next level of abstraction of the quasispecies model (a somewhat relevant discussion of the original Eigen model can be found in \cite{dress1988evolution,rumschitzki1987spectral}). This section provides a concise description of such generalization. While we concentrate here on the mathematical development of the model, we would like to note that an abstract construction of a simplicial fitness landscape can be used to model real biological systems, in particular the switching of the antigenic variants of some bacteria \cite{avery2006microbial}.

\subsection{Groups of isometries and a generalized algebraic Eigen quasispecies problem}\label{sec:6:1}
The previous results, when we encode individuals of a population by
the vertices of the binary cube $X=\{0,1\}^N$ equipped with the
Hamming distance, can be generalized as follows. Let $(X,d)$ be a
finite metric space. We will assume that the metric $d\colon X\times X
\longrightarrow \N_0$ is an {\it integer}-valued function.

Consider a group $\Gamma\leqslant{\rm Iso}(X)$ of  isometries of
$X$ and suppose that $\Gamma$ acts {\it transitively} on $X$, that
is, $X$ is a single $\Gamma$-orbit (we use the notation for the
left action). Since $\Gamma$ acts transitively on $X$ we can fix
an arbitrary point $x_0\in X$ and consider the function
$d_{x_0}\colon X\longrightarrow \N_0$ such that
$d_{x_0}(x)=d(x,x_0)$. By definition,
$$\diam (X):=\max\{d_{x_0}(x)\,|\,x\in X\}$$
 is called the {\it diameter}
of $X$. The number $N=\diam(X)$ does not depend on the choice
of $x_0$.

Let us point out a few of important general geometric examples.

\begin{example}[Weyl chamber systems] Let $\Gamma=W$ be the Weyl group of
the root system $\Delta$ of a simple finite-dimensional Lie
algebra ${\mathfrak g}$ over $\bf C$ acting on the Weyl
chamber system $X$ (see \cite{bourbaki22001lie}, chapter VI). For instance, if
$\Delta$ is of type $A_N$ then $W\cong S_N$. The distance between
two chambers $x$, $y$ is the minimal number of chamber walls  we
need to pass from $x$ to $y$. It is known (e.g., see \cite{bourbaki22001lie}, chapters
IV, V) that $d(x,y)$ is just the length of the unique element
$w\in W$ such that $y=wx$ when $W$ is viewed as a reflection group
generated by a set $S$ of reflections which correspond to
fundamental roots (see more general Example 6.3 below.) The number
$N=\diam (X)$ is known as the Coxeter number of $W$ and
$|X|=|W|$.

On the other hand, the Weyl group $W\cong S_N$ of type $A_N$ acts
on the $N$-dimensional regular simplex, the Weyl group $W$ of type
$B_N$ (or $C_N$) acts on the $N$-dimensional cube since the root
lattice is cubic in the latter case. Thus, we come to  the next
class of geometric examples.
\end{example}

\begin{example}[Regular polytopes] Let $X=P^{(0)}$ be the the set of vertices of
an $n$-dimensional regular polytope $P$ (see, e.g., \cite{coxeter1973regular} and Fig. \ref{fig:6}), all edges of
which have an integer length $e$, say, of a regular $m$-gon ($m\ge
2$) on the plane, of a tetrahedron, octahedron or icosahedron in
the 3-dimensional space (see Fig. \ref{fig:6} for some examples) and so on, equipped with the ``edge''
metric: the distance between $x$ and $y$ is the minimal number of
edges of $P$ connecting $x$ and $y$ multiplied by $e$. For
$n$-dimensional unit cube the edge metric  is the same as the
Hamming metric.

The group of all isometries $\Gamma={\rm Iso}(P)$  acts on $P$
and, consequently, on $X$. For instance, let $P$ be an icosahedron
or dodecahedron. Then $\Gamma\cong A_5$ where $A_5<S_5$ is the
alternating group of order 60.
\begin{figure}
\centering
\includegraphics[width=0.85\textwidth]{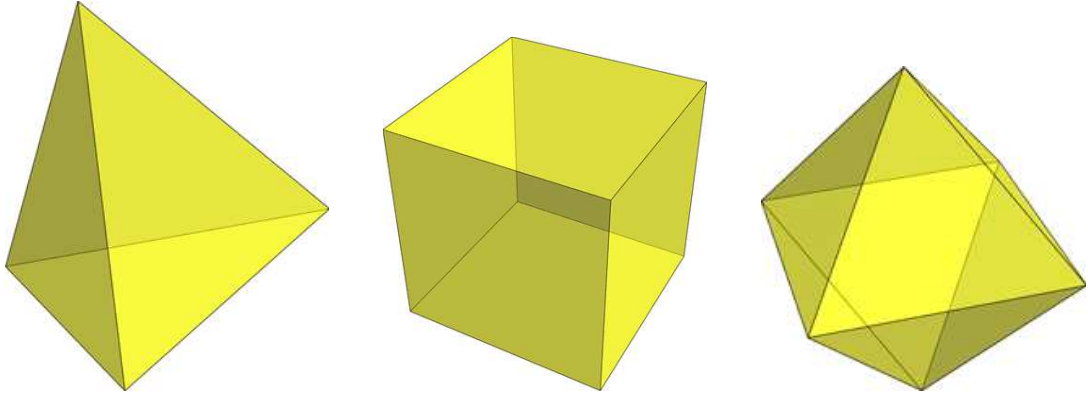}
\caption{Examples of regular polytopes in dimension 3: tetrahedron (regular simplex),
cube, octahedron}\label{fig:6}
\end{figure}
\end{example}

\begin{example}[Groups as metric spaces]\label{ex:6:3} Let $G$ be a finite group generated by a set
$S=S^{-1}$. The {\it word} metric $d=d_S$ on $G$ is defined as
follows (see \cite{de2000topics}, chapter IV for more details and examples):
$d(g,h)=l(g^{-1}h)$ where $l(g^{-1}h)=l$ is the minimal number of
generators $s\in S$ needed to represent $g^{-1}h$ as a product
$s_1\dots s_l$. The word metric is  invariant with respect to the
action of $G$ on itself by left shifts $h\to gh$. Hence, we have
the metric space $X=G$ and the transitive action of $\Gamma=G$ on
$X$ by isometries.

More generally, for any subgroup $H\leqslant G$ we can define the
metric space $X_H=\{gH\,|\,g\in G\}$ of the left cosets of $G$ by
$H$. The group $G$ acts on $X_H$ by left shifts and
$$d(gH,aH)=\min\{d(x,y)\,|\,x\in gH, y\in aH\}\,.
$$

If $G$ acts transitively by isometries on a metric space $X$ then
as a $G$-set $X$ is isomorphic to the set of left cosets $G/{\rm
St}_{\Gamma}(x_0)$, $x_0\in X$.
\end{example}
\begin{example}[$p$-adic metric spaces] Let $p$ be any fixed prime,
${\Z}_p$ be the commutative ring of $p$-adic integers equipped
with the standard $p$-adic metrics $d_p(x,y)=\|
x-y\|_p$. Consider the quotient rings
$X_{n,p}={\Z}_p/p^{n}{\Z}_p$, $n\in \N$, with the scaled metric
$\overline{d}_p(\overline{x},\overline{y})=p^{n-1}\|
x-y\|_p$ ($\overline{x}$ denotes the coset $x+p^{n}{\Z}_p$)
on which the additive group $\Gamma=\Gamma_{n,p}=X_{n,p}$ acts
1-transitively by isometries $L_{\overline{\gamma}}:
\overline{x}\to \overline{\gamma} +\overline{x}$. Here
$N_{n,p}=\diam (X_{n,p})=p^{n-1}$, $l_{n,p}=|X_{n,p}|=p^n$.

For $p=2$, $n=3$ we have 2-adic ``cube'' $X_{3,2}$ which is different from
the binary cube with the Hamming metric.
\end{example}
Now  consider a quadruple $(X,d,\Gamma, \bs w)$ where $(X,d)$ is a
finite metric space with integer distances between points of
diameter $N$ and cardinality $l=|X|$, a group $\Gamma\leqslant{\rm
Iso} (X)$ is a fixed group  and a {\it fitness landscape} $\bs
w=(w_x)^\top$ is a vector-column of non-negative real numbers
called {\it fitnesses} indexed by $x\in X$. The quadruple
$(X,d,\Gamma, \bs w)$ will be called {\it homogeneous}
$\Gamma$-landscape. In other words, the sequences of the
population are encoded by $x\in X$.

Consider also the diagonal matrix ${\bs W}={\diag}(w_x)$ of
order $l$ called the {\it fitness matrix}, the symmetric distance
matrix ${\bs D}=\bigl(d(x,y)\bigr)_{l\times l}$ with integer entries of the same order
 and the symmetric matrix
${\bs Q}=\left((1-q)^{d(x,y)} q^{N-d(x,y)}\right)_{l\times l}$ for $q\in
[0,1]$. Finally, we introduce the {\it distance polynomial}
\begin{equation}\label{6.1}
P_X(q)=\sum_{x\in X} (1-q)^{d(x,x_0)}q^{N-d(x,x_0)}\,,\quad
x_0\in X. \end{equation}

Since $\Gamma$ acts transitively on $X$ this polynomial is
independent on the choice of $x_0\in X$ and is the sum of entries
in each row (column)  of ${\bs Q}$.

The following
definition generalizes the classical Eigen quasispecies problem we dealt with
in the previous sections.

\begin{definition} The  problem
to find the leading eigenvalue $\overline{w}=\overline{w}(q)$ of
the matrix $\frac{1}{P_X(q)}{\bs{ QW}}$ and the eigenvector $\bs p=\bs p(q)$
satisfying
\begin{equation}\label{6.2}
{\bs{QWp}}=P_X(q)\overline{w}\,\bs p,\;\quad p_x=p_x(q)\ge
0,\quad\sum_{x\in X} p_x(q)=1\, \end{equation} will be called {\it
the generalized algebraic Eigen quasispecies problem}.
\end{definition}
Note that in ({\ref{6.2}})
\begin{equation}\label{6.3} \overline{w}=\sum_{x\in X} w_xp_x\,.
\end{equation} Due to the Perron--Frobenius theorem a solution of this
problem always exists. Also note that the uniform distribution
vector
\begin{equation}\label{6.4}
\bs p=\frac{1}{|X|}(1,\dots,1)^\top=\frac{1}{l}(1,\dots,1)^\top
\end{equation} provides a solution to (\ref{6.2}) in the case
of the constant fitnesses $w_x\equiv w>0$.

 The problem (\ref{6.2}) turns into the classical Eigen evolutionary problem for the $N$-dimensional binary
cube $X=\{0,1\}^N$ with the Hamming metric and $\Gamma={\rm Iso}
(X)$ which was named in 1930 by A. Young a {\it hyperoctahedral}
group. $\Gamma$ is isomorphic as an abstract group to the Weyl
group of the root system of type $B_N$ or $C_N$ and is acting on
the cube. In the classical case $P_X(q)\equiv 1$.

Consider also the following serial examples. If $X$ is the set of
vertices of an $n$-dimensional regular simplex with all edges of
unit length then $\Gamma={\rm Iso} (X)\cong S_{n+1}$, $N=\diam(X)=1$ and $l=|X|=n+1$. The distance polynomial is
\begin{equation}\label{6.5}
P_X(q)=q+n(1-q)\,.
\end{equation}
If $X$ is the set of vertices of an $n$-dimensional
hyperoctahedron with all edges of unit length then $\Gamma={\rm
Iso} (X)$ is again a hyperoctahedral group (the hyperoctahedron is
the dual polytope to the cube), $N=\diam(X)=2$ and
$l=|X|=2n$. The distance polynomial is
\begin{equation}\label{6.6}
P_X(q)=q^2+(2n-2)(1-q)q+(1-q)^2\,.
\end{equation}

\subsection{Properties of the distance polynomial}

 In the notation of Section \ref{sec:6:1}
consider the polynomial $P_X(q)=P_{X,d}(q)$. The polynomial
$P_X(q)$ is strictly positive on $[0,1]$ (if $N$ is strictly equal
to $\diam(X)$. If  $N>\diam(X)$ then $P_X(0)=0$, such cases
sometimes we will need to  consider) and  possesses the following
properties:
\begin{enumerate}
\item
\begin{equation}\label{6.7}
P_X(1)=1,\quad P_X\left(\frac{1}{2}\right)=
\frac{|X|}{2^N}=\frac{l}{2^N}\;. \end{equation}

\item\begin{equation}\label{6.8}
P_X(q)=\sum_{k=0}^N f_k\,(1-q)^k q^{N-k}\in \mathbf Z[q]\,,
\end{equation} where the non-negative integers $f_k=f_k(X):=\#\{x\in X\,|\,d(x,x_0)=k\}$ are the cardinalities of
$d$-spheres in $X$ with the center at the fixed point $x_0$ and of
radius $k$.
\end{enumerate}

\begin{remark} The polynomial $S_X(t)=\sum_{k=0}^N f_k
t^k$ is often  called the {\it spherical growth function} of
$(X,d)$. See, for instance, \cite{de2000topics}, chapter VI for details and
examples.
\end{remark}

Suppose that we scaled the metric $d$ by a positive integer factor
$e$. Let $P_{X,e\cdot d}(q)$ denote the new distance polynomial.
Then
\begin{equation}\label{6.9} P_{X,e\cdot d}(q)=\sum_{x\in X}
(1-q)^{e\,d(x,x_0)}q^{e(N-d(x,x_0))}=\sum_{k=0}^N
f_k\,(1-q)^{e\,k} q^{e(N-k)}\,,\quad x_0\in X\,.
\end{equation}
Since $q\in [0,1]$ we may assert that the sequence $\{P_{X,e\cdot
d}(q)\,|\,e\in \N\}$ is non-increasing at each fixed point
$q\in[0,1]$.

\subsection{Regular simplicial fitness landscapes}
To give a specific example of the analysis of the generalized algebraic Eigen quasispecies problem we shall briefly consider
two-valued fitness landscapes related to the set of vertices of the regular
$n$-dimensional simplex $X$ with ${\rm Iso}(X)\cong S_{n+1}$. Here we follow the main lines of Section \ref{sec:1}.

Biologically, the simplicial fitness landscape means that we deal with a population of individuals such that any individual can mutate to any other individual with the same probability equal to $1-q$. Even such oversimplified construction can model a non-trivial biological system. Here, for example, if we consider ``mutation'' as a sudden discrete genetic (heritable) change then the simplicial geometry can describe, at a first approximation, the switching of the antigenic variants for some bacteria. These variants turns one into another with almost equal probabilities, whereas the corresponding fitnesses of different variants are defined by interactions with the host immune system (e.g., \cite{avery2006microbial}).

\subsubsection{General scheme}
Let $X=\{0,1,\dots,n\}$ and $d(i,j)=1$ if $i\ne j$, $d(i,i)=0$.
Hence, $X$ is a metric space with the trivial metric, $N=\diam(X)=1$ and $l=|X|=n+1$. The distance polynomial is defined by
(\ref{6.5}).

Let $A\subset\{0,1,\dots,n\}$. Consider the landscape
$$
w_k=\left\{
\begin{array}{r}
w+s,\quad k\in A\;,\\
w,\quad k\notin A\;.\\
\end{array}
\right.
$$
The matrix ${\bs W}$ of fitnesses can be represented as follows
\begin{equation}{\label{6.10}}
{\bs W}=w{\bs I}+s {\bs E}_A=w {\bs I}+s
\sum_{a\in A}{\bs E}_a,
\end{equation}
$\bs{ I}$ being the identity matrix and ${\bs E}_a$ being the
elementary matrix with the only one nontrivial entry $e_{aa}=1$ on
the diagonal.

We want to solve the problem (\ref{6.2}). The matrix ${\bs
Q}=({\bs Q}_{ba})=(2q-1){\bs I}+(1-q){\bs E}$ where all the
entries of $\bs{ E}$ are ones, that is
$$
{\bs Q}_{ba}=\left\{
\begin{array}{cc}
1-q\,,& b\ne a\,,\\
q\,,&b=a\,.
\end{array}
 \right.
$$

It can be directly checked
that
\begin{equation}{\label{6.11}}
{\bs{D}}(n,q):={\bs T}^{-1} {\bs{QT}}=\diag(q+n(1-q),2q-1,\dots,2q-1),
\end{equation}
where for the symmetric transition matrix ${\bs T}$ of order $n+1$
we have
\begin{equation}{\label{6.12}}
{\bs T}=\left(\begin{array}{crrrr} 1&1&1&\dots&1\\
1&-1&0&\dots&0\\
1&0&-1&\dots&0\\
\vdots&\vdots&\vdots&\ddots&0\\
1&0&0&\dots&-1\\
\end{array}\right)\,,\quad {\bs T}^{-1}=\frac{1}{n+1}\left(\begin{array}{crrrr} 1&1&1&\dots&1\\
1&-n&1&\dots&1\\
1&1&-n&\dots&1\\
\vdots&\vdots&\vdots&\ddots&1\\
1&1&1&\dots&-n\\
\end{array}\right)\,.
\end{equation}

\medskip The transformation of (\ref{6.2}) yields
$$
{\bs T}^{-1}{\bs{ QT}} {\bs T}^{-1}{\bs{ WT}} {\bs
T}^{-1}\bs p=P_X(q)\overline{w}\, {\bs T}^{-1}\bs p=(q+n(1-q)) \overline{w}\,
{\bs T}^{-1}\bs p.
$$
or, in view of (\ref{6.11}),
\begin{equation}{\label{6.13}}
{\bs D}(n,q) \left(w{\bs I}+s\sum_{a\in A}{\bs T}^{-1}{\bs
E}_a{\bs T}\right) {\bs T}^{-1}\bs p=P_X(q)\overline{w}\, {\bs
T}^{-1}\bs p,
\end{equation}
whence
\begin{equation}{\label{6.14}}(P_X(q)\overline{w}{\bs
I}-w{\bs D}(n,q)) {\bs T}^{-1}\bs p=\sum_{a\in A} {\bs
D}(n,q){\bs T}^{-1}{\bs E}_a\bs p.
\end{equation}

Let $\bs x={\bs T}^{-1}\bs p$, $\bs p={\bs T}\bs x$. Then (\ref{6.14}) implies
\begin{equation}{\label{6.15}} \bs x=s\sum_{a\in A}(P_X(q)\overline{w}\,{\bs I}-{\bs D}(n,q)))^{-1}{\bs D}(n,q){\bs T}^{-1}{\bs E}_a\bs p,\end{equation}
 or, in
coordinates,
\begin{equation}{\label{6.16}}
x_k=s\sum_{a\in A}\frac{{\bs D}(n,q)_k
t^{(-1)}_{ka}p_a}{P_X(q)\overline{w}-w{\bs D}(n,q)_k},\quad
k=0,\ldots,n.
\end{equation}

Since $\bs p={\bs T}\bs x$, then we get from (\ref{6.15})
\begin{equation}{\label{6.17}}
\bs p=s\sum_{a\in A}{\bs T}(P_X(q)\overline{w}{\bs I}-{\bs
D}(n,q))^{-1} { \bs D}(n,q){\bs T}^{-1}\,{\bs E}_a\bs p.
\end{equation}

Only the components $p_a$,
 $a\in A$, are involved in the right-hand side of (\ref{6.17}).
By definition, ${\bs E}_A=\sum_{a\in A} {\bs E}_a$ and ${\bs E}_A$
is a projection matrix. Hence, we can multiply both sides of
(\ref{6.17}) by ${\bs E}_A$:
\begin{equation}{\label{6.18}}
{\bs E}_A \bs p=s{\bs E}_A{\bs T}(P_X(q)\overline{w}{\bs I}-{\bs
D}(n,q)))^{-1} {\bs D}(n,q){\bs T}^{-1}\,{\bs E}_A \bs p.
\end{equation}

We can rewrite (\ref{6.18}) as
\begin{equation}{\label{6.19}}
\bs p_A={\bs M}\bs  p_A,\quad \bs p_A={\bs E}_A \bs p\,,
\end{equation}
where
\begin{equation}{\label{6.20}}
\begin{split}
{\bs M}&=s{\bs E}_A \,{\bs T}(P_X(q)\overline{w}{\bs I}-{\bs D}(n,q)))^{-1}{\bs D}(n,q){\bs T}^{-1}\\
&=\frac{s}{\overline{w}P_X(q)}\sum_{c=0}^{\infty}\left(\frac{w}{\overline{w}P_X(q)}\right)^c{\bs E}_A \,{\bs T}{\bs D}(n,q)^{c+1}{\bs T}^{-1}\\
&=\frac{s}{\overline{w}P_X(q)}\sum_{c=0}^{\infty}\left(\frac{w}{\overline{w}P_X(q)}\right)^c{\bs E}_A \,{\bs Q}^{c+1}\,.
\end{split}
\end{equation}

It follows that vector $\bs p_A$ is an
eigenvector of ${\bs M}$ corresponding to the eigenvalue
$\lambda=1$.

Consider $\overline{w}$ in (\ref{6.19}), (\ref{6.20}) as a
parameter. It follows from (\ref{6.3}) that $\overline{w}$ depends
only on $p_a$, $a\in A$, that is, on the ``reduced'' vector
$\bs p_A={\bs E}_A\bs p$. The original eigenvector $\bs p$ can be
reconstructed from $\bs p_A$ with the help of (\ref{6.17}). Thus, instead of the original problem we arrive to the reduced problem to find the eigenvector
$\bs p_A$ satisfying (\ref{6.18}) and corresponding to the eigenvalue
$\lambda=1$ of the matrix ${\bs M}=(m_{ba})$ defined in (\ref{6.20})).
The parameter $\overline{w}=\overline{w}(q)$ satisfies the formula
\begin{equation}{\label{6.21}}
\overline{w}=w+s\sum_{a\in A}p_a. \end{equation}

\subsubsection{Equation for the eigenvalue $\overline{w}$}

\begin{lemma}In \eqref{6.20} we have
\begin{equation}{\label{6.22}}
{\bs Q}^{c+1}=(2q-1)^{c+1}{\bs I}+\sum_{m=1}^{c+1}{\binom{c+1}{m}}(n+1)^{m-1}(2q-1)^{c+1-m}(1-q)^m{\bs E}\,.
\end{equation}
\end{lemma}
\begin{proof} We apply directly the binomial
expansion for the matrix ${\bs Q}^{c+1}=((2q-1){\bs I}+(1-q){\bs
E})^{c+1}$. Since all the entries of $\bs{ E}$ are ones then
$\bs E^2=(n+1)\bs E$ and, consequently, $\bs E^m=(n+1)^{m-1}\bs E$.
\end{proof}

The equality (\ref{6.19}) implies that
\begin{equation}{\label{6.23}}
\sum_{b\in A} p_b= \sum_{b\in A}\sum_{a\in A} m_{ba}p_a=\sum_{a\in
A}p_a \sum_{b\in A} m_{ba}. \end{equation}
Suppose that the sum $\sum\limits_{b\in A}m_{ba}$ does not depend
on $a\in A$. Then it follows from (\ref{6.23}) that
$\sum\limits_{b\in A} m_{ba}=1$ for each $a\in A$. In view of
(\ref{6.20})
\begin{equation}{\label{6.24}}
\frac{\overline{w}P_X(q)}{s}=
\sum_{c=0}^{\infty}\left(\frac{w}{\overline{w}P_X(q)}\right)^c
\sum_{b\in A}({\bs E}_A \,{\bs Q}^{c+1})_{ba}\,,
\end{equation}
if the inner sum does not depend on $a\in A$.

\begin{theorem} In the previous notation let $A$ be a subset
 of a simplicial metric space $X$. Then the equality
\begin{equation}{\label{6.25}}
\frac{\overline{w}P_X(q)}{s}=
\!\sum_{c=0}^{\infty}\left(\frac{w}{\overline{w}P_X(q)}\right)^c
\left(\!(2q-1)^{c+1}+\frac{|A|}{n+1}\sum_{m=1}^{c+1}{\binom{c+1}{m}}(n+1)^{m}(2q-1)^{c+1-m}(1-q)^m\right)\,
\end{equation}
holds.
\end{theorem}
\begin{proof} Since ${\rm Iso}(X)\cong S_{n+1}$ acts
$(n+1)$-transitively on $X$ we may assume that $A$ is the subset
$\{0,1,\dots, |A|-1\}$ on which the cyclic subgroup
$C_{|A|}=\left<(0,1,\dots, |A|-1)\right>$ is acting. Then we apply
(\ref{6.22}) to (\ref{6.24}) . \end{proof}

The formula (\ref{6.25}) can be simplified as follows. Recall (see
(\ref{6.5})) that $P_X(q)=q+n(1-q)$. The binomial expansion yields
$$\sum_{m=1}^{c+1}{\binom {c+1}{m}}(n+1)^{m}(2q-1)^{c+1-m}(1-q)^m$$
$$=
((2q-1)+(n+1)(1-q))^{c+1}-(2q-1)^{c+1}=P_X(q)^{c+1}-(2q-1)^{c+1}\,.$$
Hence, (\ref{6.25}) reads
\begin{equation}{\label{6.26}}
\frac{\overline{w}P_X(q)}{s}=
\!\sum_{c=0}^{\infty}\left(\frac{w}{\overline{w}P_X(q)}\right)^c
\left(\!\left(1-\frac{|A|}{n+1}\right)(2q-1)^{c+1}+\frac{|A|}{n+1}P_X(q)^{c+1}\right).
\end{equation}

Summing the geometric progressions we finally get
\begin{equation}{\label{6.27}}
\frac{|A|}{(n+1)(\overline{u}-u)}+
\left(1-\frac{|A|}{n+1}\right)\frac{2q-1}{(q+n(1-q))\overline{u}-(2q-1)u}=1,\quad
u=\frac{w}{s}\,,\quad \overline{u}=\frac{\overline{w}}{s}\,.
\end{equation}

\begin{remark} Note that the equation (\ref{6.27}) depends only on $|A|$
and the dimension $n$.  It follows that (\ref{6.27}) provides the
eigenvalue of the two-valued fitness problem (\ref{6.2}) for any
subset $A\subset X$. Note also that the equation (\ref{6.27})
turns into the equation of degree $2=N+1$ where $N= 1$ is the
diameter of the simplex (compare with Corollary \ref{corr:2:1}).
We expect that for the hyperoctahedral landscapes we will get
cubic equations since $N=2$ for a hyperoctahedron (with unit
edges) in any dimension.

The solution to (\ref{6.27}) is given by the following formula\footnote{This is a correct formula, unfortunately, in the published version the factor $(q+n(1-q))$ is missing}
\begin{equation}{\label{6.28}}
\begin{split}
\overline{u}&=\frac{v(q)+\sqrt{v^2(q)-4(u+u^2)(2q-1)(q+n(1-q))}}{2(q+n(1-q))}\,,\\
v(q)&=(q+n(1-q))\left(u+\frac{|A|}{n+1}\right)+(2q-1)\left(u+1-\frac{|A|}{n+1}\right).
\end{split}
\end{equation}
\end{remark}
\subsubsection{Simplicial error threshold}

In this subsection some results of Section \ref{sec:5} are appropriated for
the case of the simplicial landscapes.

Let $X_n$ be the set of vertices of an $n$-dimensional regular
simplex with edges of unit length and let $(A_n)_{n=n_0}^\infty$, $A_n\subset
X_n$, be a sequence of subsets. Let
$\overline{u}=\overline{u}^{(n)}(q)$ be the sequence of
the corresponding eigenvalues (see (\ref{6.28})). It can be checked
that each function $\overline{u}^{(n)}$ is increasing on the
segment $[0.5,1]$ and convex downward.

A sequence $(A_n)_{n\geq n_0}$ is called a sequence of the {\it
moderate growth} if
\begin{equation}{\label{6.29}}
\lim_{n\to \infty}\frac{|A_n|}{n+1}=0\;.
\end{equation}

Let us denote $\alpha_n=\frac{|A_n|}{n+1}$. In view of
(\ref{6.28}) $\overline{u}^{(n)}(0.5)=u+\alpha_n\to u$ as
$n\to\infty$ if the sequence $(A_n)_{n\geq n_0}$ is of the moderate growth. On
the other hand, $\overline{u}^{(n)}(1)\equiv u+1$.

Consider new coordinates $x$, $L$ such that
\begin{equation}{\label{6.30}}q=1-\frac{x}{n},\;0\leq x\leq \frac{n}{2},
\quad L=\frac{u+1}{\overline{u}}-1\,,
 \;0\leq L\leq 1/u.\end{equation}
We assume that $u>0$ in (\ref{6.30}). The curve
$\overline{u}=\overline{u}^{(n)}(q)$ transforms into the curve
\begin{equation}{\label{6.31}}
L=L_n(x)=\frac{u+1}{\overline{u}^{(n)}\left(
1-\frac{x}{n}\right)}-1.
\end{equation}

\begin{definition}\label{def:6:1} We say that a sequence $(A_n)_{n=n_0}^\infty$
of the moderate growth,  or, equivalently, the family
$(\overline{u}^{(n)})_{n=n_0}^\infty$ possesses the threshold-like behavior on the segment
$[0.5,1]$ if  for each fixed $x\geq 0$ and the corresponding
functions $L_n(x)$ it it true that
\begin{equation}{\label{6.32}}
\lim_{n\to\infty} L_n(x)=L(x)=\left\{\begin{array}{l}\vspace{3pt}
\;\displaystyle x,\quad 0\leq x< \frac{1}{u}\,,\\
\displaystyle \frac{1}{u},\quad x\geq \frac{1}{u}\;.
\end{array}
\right.
\end{equation}
\end{definition}

%\begin{picture}(400,130)
%
%\put(130,40){{\vector(1,0){200}}} \put(140,30){{\vector(0,1){90}}}
%\put(140,40){{\circle*{2}}} \put(140,90){{\circle*{2}}}
%\put(190,40){{\circle*{2}}} \put(120,20){0}
%
%\put(330,25){$x$} \put(125,110){$L$} \put(125,85){$\frac{1}{u}$}
%\put(188,25){$\frac{1}{u}$} \put(200,70){$L=L_n(x)$}
%\put(180,100){$L=L(x)$}
%
%\qbezier(220,89)(208,89)(191,87)\qbezier(140,40)(182,85)(191,87)
%
%\linethickness{0.7pt}{\put(190,90){{\line(1,0){120}}}
%\put(140,40){{\line(1,1){50}}}  }
% \put(100,5){Figure$^{****}$. Simplicial error threshold in coordinates $x$, $L$}
%\end{picture}

\begin{figure}
\centering
\includegraphics[width=0.5\textwidth]{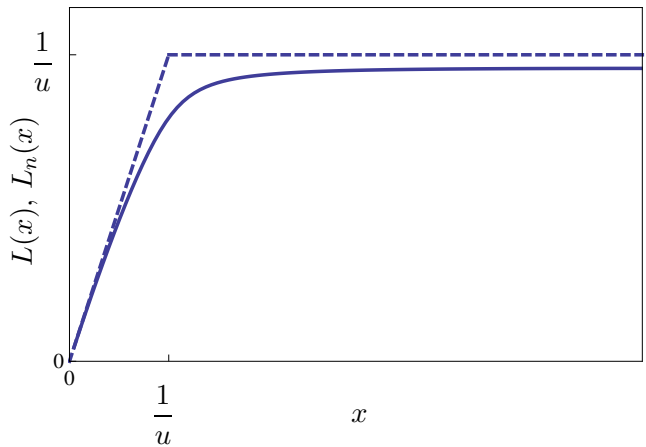}
\caption{Illustration to Definition \ref{def:6:1}}
\end{figure}

The following formula provides an approximation for the error
threshold value (if exists) $q_*^{(n)}(u)$, $n\gg 1$:
$$
q^{(n)}_*(u)\approx 1-\frac{1}{nu}=1-\frac{s}{nw}\,.
$$

\begin{theorem}In the above notation suppose that a sequence
of subsets $A_n\subset X_n$  of the moderate growth is given and
$u>0$. Then the sequence $(A_n)_{n=n_0}^\infty$ shows the threshold-like behavior on the segment
$[0.5,1]$.
\end{theorem}
\begin{proof}[Sketch of a proof] In view of the equation (\ref{6.27}),
in coordinates $x$, $L$:
\begin{equation}{\label{6.33}}
 \frac{\alpha_n(1+L_n(x))}{1-uL_n(x)}+
\frac{(1-\alpha_n)(1-2x/n)}{(1+x-x/n)\frac{1+u}{1+L_n(x)}-\left(1-2x/n\right)u}=1\;.
\end{equation}

The existence of $\lim\limits_{n\to\infty}L_n(x)$ for a fixed $x$,
$0\leq x<1/u$, can be proved with the help of lower and upper
estimates. If $n\to\infty$ in (\ref{6.33}) we get (since
$\alpha_n=|A_n|/(n+1)\to 0$)
$$
\frac{1}{(1+x)\frac{1+u}{1+L(x)}-u}=1\,,
$$
or $L(x)\equiv x$ on $[0,1/u)$. Since $L_n(x)$ increases with
respect to $n$ and cannot exceed the value $1/u$ we obtain the
desired result.
\end{proof}

In a similar fashion other geometric examples can be analyzed.

\section{Concluding remarks}
There are two main points to emphasize in order to conclude the presentation. First, in this text we put forward general, rigorous, and quite elementary methods to analyze two-valued fitness landscapes in the classical quasispecies model. While a great deal of analysis of this problem in the existing literature was inspired by the analogies with the famous Ising model of statistical physics, we show that direct methods of linear algebra allow gaining full understanding of the properties of the selection--mutation equilibrium in this model at least in some special cases.

Second, the language of the group theory gives us an opportunity
to look at the phenomena associated with the quasispecies model
from a more general and abstract point of view. In particular, the
infamous error threshold can be looked at from the position of the
external and internal metric properties of orbits. If the set of
population sequences is enumerated by points of a finite metric
space $X$ with integer-valued metric $d$ on which a group $\Gamma$
acts transitively by isometries then we can involve  group
theoretical and algebraic tools in order to obtain not very
complicated solutions for the leading eigenvalue problem, at least
in the special case of the two-valued fitness landscapes. Such a
classical approach is in accordance with the well known F. Klein's
{\it Erlangen program}. We are convinced that this connection
between mathematical biology, finite geometries, combinatorics and
algebra confirms the importance of Eigen's model from various
viewpoints.

To reiterate, in the general case we consider a quadruple  $(X,d,\Gamma, \bs w)$
--- {\it homogeneous} $\Gamma$-landscape ---  with the fitness
function $\bs w\colon X\longrightarrow \R_{\geq 0}$. The information of the
geometric properties of the underlying metric space $(X,d)$ is
contained in
 the symmetric matrix
${\bs Q}=\bigl((1-q)^{d(x,y)} q^{N-d(x,y)}\bigr)$, $q\in
[0,1]$, where $N=\diam(X)$ is the diameter of $X$. The
diameter $N$ as well as the cardinality $l=|X|$ are the two main
numerical characteristics of the model $(X,d,\Gamma, \bs w)$. The distance polynomial $P_X(q)$, which is the leading
eigenvalue of the matrix ${\bs Q}$, plays the key role in the analysis. For the classical Eigen's
quasispecies model $X=\{0,1\}^N$ is the binary cube with the Hamming distance,
$\diam(X)=N$, $l=|X|=2^N$ and $P_X(q)\equiv 1$.

Suppose that we have a subgroup $G\leqslant\Gamma$ which also acts
on $X$ and the fitness function $\bs w$ is constant on the orbits
of $G$-action. We saw in Sections \ref{sec:2}, \ref{sec:6} that
for the two-valued fitness functions $\bs w(A)=w+s$, $\bs
w(X\setminus A)=w$, $A$ being any $G$-orbit, the degree of the
equation on the leading eigenvalue can be reduced from $l$ to
$N+1$. Although the solution of the leading eigenvalue problem
appears in an implicit form we are able to obtain lower and upper
bounds for it.

We can also consider sequences of metric spaces $X_n$ and
orbits $A_n$ as $n\to\infty$. Usually we have a chain
$$
X_{n_0}\subset\dots \subset X_{n}\subset X_{n+1}\subset\dots
\subset \bigcup_{n} X_n=X_{\infty}\,.
$$
The analysis presented in the main text allows to conjecture that the {\it error threshold}, i.e., non-analytical behavior of the leading eigenvalue $\overline{w}$ in the infinite sequence limit,  occurs when
the cardinalities $|A_n|$ grow not rapidly enough comparing with the growth of $|X_n|$. For instance, when $A_n\equiv A$, where $A$
contains a single point (the single peaked landscape), or is a fixed constant set
(orbit) then the threshold-like behavior is observed. This topic is to be
investigated in the general situation.

At the beginning of Section \ref{sec:6} we pointed out the most interesting
geometric examples of groups and metric spaces for which the
generalized Eigen's algebraic problem could be solved. Among them are the Weyl groups acting on the chamber systems
(the {\it reflection} groups should be added) and groups of
symmetry of regular polytopes. Example \ref{ex:6:3} deals with all
 finite groups in general. It is very possible, and genuinely intriguing, that some
infinite finitely generated groups (free groups, non-Euclidean
crystallographic groups and others) can be included in the list of groups for the
future research (see, for instance, \cite{de2000topics,gromov1993}).

\appendix
\section{Proof of Proposition \ref{pr:ad}}\label{ap:1}
The following three lemmas and corollary provide the full proof that all the
examples in Section \ref{sec:3} deal with admissible sequences of orbits of the
moderate growth (Proposition \ref{pr:ad}).

\begin{lemma} Let $A\subset X_{n_0}$ be a fixed
$G$-orbit. Consider the constant sequences $A_n\equiv A$ and
$G_n\equiv G$, $n\geq n_0$. Then the sequence $(A_n)_{n=n_0}^\infty$ is
admissible.
\end{lemma}

\begin{proof}Since the orbit is not changing as
$n\to \infty$ then it follows from \eqref{eq:23} that
$$
F_{A_n}(2q-1)=q^{n-n_0}F_{A_{n_0}}(2q-1)\,,\quad q\in[0,1].
$$
The polynomial $F_{A_n}(2q-1)>0$  and $q^{n-n_0}\geq q^{n+1-n_0}$
on $[0.5,1]$. Hence, \eqref{eq:61} holds. \end{proof}

\begin{lemma} Let $a_n\in X_n$, $a_n^*=2^n-1-a_n$,
and $A_n=\{a_n,a_n^*\}$. Let $G_n=G=\{1,g\}$ be the group of order
2 such that $g(a)=a^*$ for any $a\in X_n$. Then the sequence
$(A_n)_{n=n_0}^\infty$ is admissible.
\end{lemma}

\begin{proof} In view of \eqref{eq:27} and \eqref{eq:49}
$$
F_{A_n}(2q-1)=q^n+(1-q)^n\geq
q^{n+1}+(1-q)^{n+1}=F_{A_{n+1}}(2q-1)\,,\quad q\in[0,1].
$$
\end{proof}

\begin{lemma}\label{l:5:3} Let  $p$ be a fixed number, $n\geq n_0=2p$.
Let  $A_n=A_{n,p}=\{a\in X_n\,|\,H_a=p\}$. Then $(A_n)_{n=n_0}^\infty$ is an
admissible sequence.
\end{lemma}

\begin{proof} It follows from \eqref{eq:45} that
\begin{equation}\label{eq:63}
F_{A_n}(2q-1)=\sum\limits_{k=0}^p {\binom{p}{k}}{\binom{n-p}{k}}(1-q)^{2k}q^{n-2k}=:F_{n,p}(q)\;,\quad 0\leq p\leq n.
\end{equation}
At the same time  consider the polynomials
\begin{equation}\label{eq:64}
G_{n,p}(q):=\sum\limits_{k=1}^p {\binom{p}{k}}{\binom{n-p}{k-1}}(1-q)^{2k}q^{n+1-2k}.
\end{equation}
By definition, $F_{0,0}(q)\equiv 1$,  $G_{n,0}(q)=0$. Applying the
binomial formulas ${\binom{n+1-p}{k}}={\binom{n-p}{k}}+{\binom{n-p}{k-1}}$ to \eqref{eq:63} and ${\binom{p}{k}}={\binom{p-1}{k}}+{\binom{p-1}{k-1}}$
to \eqref{eq:64} we get the following recursive relations:
\begin{equation}\label{eq:65}
F_{n+1,p}(q)=qF_{n,p}(q)+G_{n,p}(q)\,,\quad G_{n+1,p}(q)=(1-q)^2
F_{n,p-1}(q)+qG_{n,p-1}(q)\,.
\end{equation}

When we substitute the left-hand-side of the second formula \eqref{eq:65} into
the first one (with the change $n\to n-1$) and then iterate such
substitutions we get
\begin{equation}\label{eq:66}
F_{n+1,p}(q)=qF_{n,p}(q)+(1-q)^2\sum_{j=1}^p
q^{j-1}F_{n-j,p-j}(q)\;.
\end{equation}
In the same way the equality
\begin{equation}\label{eq:67}
G_{n+1,p}(q)=(1-q)^2\sum_{j=1}^p q^{j-1}F_{n+1-j,p-j}(q)
\end{equation}
can be obtained.

Formulas \eqref{eq:65} imply also that
$$
F_{n,p}(q)-F_{n+1,p}(q)=(1-q)F_{n,p}(q)-G_{n,p}(q)=
(1-q)qF_{n-1,p}(q)+(1-q)G_{n-1,p}(q)-G_{n,p}(q)\,,
$$
or
\begin{equation}\label{eq:68}
F_{n,p}(q)-F_{n+1,p}(q)=q\left((1-q)F_{n-1,p}(q)-G_{n,p}(q)\right)+
(1-q)\left(G_{n-1,p}(q)-G_{n,p}(q)\right)\,.
\end{equation}
Our objective is to prove that $$F_{2p+k+1,p}(q)\leq
F_{2p+k,p}(q)\,,\qquad k\geq 0\,,\quad q\in[0,1]\,.$$
We will proceed by induction on $p$ and, for a fixed $p$, by induction on $k$.

First of all, the case $p=0$ is trivial since $F_{n,0}(q)=q^n$.

Let $p\ge 1$ be fixed and let $k=0$. Substituting $n=2p$ into \eqref{eq:68}
we get
$$
F_{2p,p}(q)-F_{2p+1,p}(q)=q\left((1-q)F_{2p-1,p}(q)-G_{2p,p}(q)\right)+
(1-q)\left(G_{2p-1,p}(q)-G_{2p,p}(q)\right)\,.
$$
Let us show that both summands in the right-hand side are
nonnegative on $[0,1]$. On the one hand, by definition we have
$F_{n,p}(q)=F_{n,n-p}(q)$. Then in view of \eqref{eq:65}
\begin{align*}
(1-q)F_{2p-1,p}(q)-G_{2p,p}(q)&=(1-q)F_{2p-1,p-1}(q)-G_{2p,p}(q)\\
&=(1-q)F_{2p-1,p-1}(q)-((1-q)^2F_{2p-1,p-1}(q)+qG_{2p-1,p-1}(q))\\
&=q(1-q)F_{2p-1,p-1}(q)-qG_{2p-1,p-1}(q)\\
&=q(1-q)F_{2p-1,p-1}(q)-q(F_{2p,p-1}-qF_{2p-1,p-1}(q))\\
&=q(F_{2p-1,p-1}(q)-F_{2p,p-1}(q))\geq 0\;,\quad q\in [0,1]\,,
\end{align*}
by the inductive hypothesis.

On the other hand, it follows from \eqref{eq:67} that
$$
G_{2p-1,p}(q)-G_{2p,p}(q)=(1-q)^2\sum_{j=1}^p
q^{j-1}(F_{2p-j,p-j}(q)-F_{2p+1-j,p-j}(q)) \geq 0
$$
on $[0,1]$ by the same reasons. This finishes the proof for the
case $k=0$.

Let $k\ge 1$. Then by virtue of \eqref{eq:66} we can assert that
\begin{align*}
&F_{2p+k,p}(q)-F_{2p+k+1,p}(q)=\\
&=q(F_{2p+k-1,p}(q)-F_{2p+k,p}(q))+(1-q)^2\sum_{j=1}^p
q^{j-1}(F_{2p+k-1-j,p-j}(q)-F_{2p+k-j,p-j}(q))\geq 0
\end{align*}
on $[0,1]$ by the inductive hypothesis. The lemma is proved.
\end{proof}

%It should be also mentioned (we will not use this in what follows)
%that $F_{n,p-1}(q)\leq F_{n,p}(q)$ on $[0,1]$ if $2p+1\leq n$ and
%$F_{n,n-p+1}(q)\leq F_{n,n-p}(q)$, due to the symmetry
%$F_{n,p}(q)= F_{n,n-p}(q)$.

\begin{corollary} Let
$A_{2n}=A_{2n,n}=\{a\in X_{2n}\,|\,H_a=n\}$. Then $(A_{2n})_{n=n_0}^\infty$ is
an admissible sequence.
\end{corollary}

\begin{proof} In the notation of Lemma \ref{l:5:3} let us
prove that $F_{2n,n}(q)\leq F_{2n-2,n-1}(q)$ on $[0,1]$. From the
first formula \eqref{eq:65} we can find the expressions
$G_{n+1,p}(q)=F_{n+2,p}(q)-qF_{n+1,p}(q)$,
$G_{n,p-1}(q)=F_{n+1,p-1}(q)-qF_{n,p-1}(q)$ and substitute them
into the second one. The simplification yields
$$
F_{n+2,p}(q)=(1-2q)F_{n,p-1}(q)+qF_{n+1,p}(q)+qF_{n+1,p-1}(q)\,.
$$
Consequently, choosing appropriate values for $n$, $p$ in this
formula, we get
$$
F_{2n-2,n-1}(q)-F_{2n,n}(q)=q(F_{2n-2,n-1}(q)-F_{2n-1,n}(q))+
q(F_{2n-2,n-1}(q)-F_{2n-1,n-1}(q))\,.
$$
But in view of \eqref{eq:63} $F_{n,p}(q)=F_{n,n-p}(q)$ for all $n$ and $p$,
$0\leq p\leq n$. Hence $F_{2n-1,n}(q)=F_{2n-1,n-1}(q)$ and it
follows from Lemma \ref{l:5:3} that
$$
F_{2n-2,n-1}(q)-F_{2n,n}(q)=2q(F_{2n-2,n-1}(q)-F_{2n-1,n-1}(q))\geq
0\,
$$
on the segment $[0,1]$.
\end{proof}

\paragraph{Acknowledgements:}
The YSS's research is partially supported by the joint
grant between the Russian Foundation for Basic Research (RFBR) and Taiwan National Council \#12-01-92004HHC-a and by RFBR grant \#13-01-00779.
ASN's research is supported in part by ND EPSCoR and NSF grant \#EPS-0814442. We thank Yuri Wolf from NCBI/NLM/NIH for a profitable discussion on the biological examples of the generalized Eigen's models. 

%\bibliography{prebiotic}

\end{document}